\documentclass[lettersize,journal]{IEEEtran}

\usepackage{amsmath,amsfonts}
\usepackage{algorithm}
\usepackage{algorithmic}
\usepackage{array}
\usepackage[caption=false,font=normalsize,labelfont=sf,textfont=sf]{subfig}
\usepackage{textcomp}
\usepackage{stfloats}
\usepackage{url}
\usepackage{verbatim}
\usepackage{graphicx}
\usepackage{cite}
\usepackage{amssymb}
\usepackage{amsthm}
\newtheorem{theorem}{Theorem}
\newtheorem{lemma}{Lemma}

\newtheorem{corollary}{Corollary}
\newtheorem{assumption}{Assumption}
\theoremstyle{definition}
\newtheorem{definition}{Definition}
\newtheorem{remark}{Remark}

\DeclareMathOperator*{\argmax}{arg\,max}

\DeclareMathOperator*{\subjectto}{subject\,to}
\usepackage{xcolor}

\usepackage{color}

\begin{document}

\title{Truth Revelation, Information Hiding, or Misinformation: Characterization of Equilibrium Outcomes in Signaling Games}

\author{Ertan~Kaz{\i}kl{\i}, Sinan~Gezici, and~Serdar~Y{\"u}ksel
\thanks{E. Kaz{\i}kl{\i} is with the Department of Electrical and Electronics Engineering, TOBB University of Economics and Technology, 06560, Ankara, Turkey. Email: ekazikli@etu.edu.tr.\\
\indent S. Gezici is with the Department of Electrical and Electronics Engineering, Bilkent University, 06800, Ankara, Turkey, Email: gezici@ee.bilkent.edu.tr.\\
\indent S. Y\"uksel is with the Department of Mathematics and Statistics, Queen's University, K7L 3N6, Kingston, Ontario, Canada. Email: yuksel@mast.queensu.ca.
}
}



\maketitle

\begin{abstract}
In signaling games where a sender and a receiver have misaligned criteria, equilibrium behavior may lead to fully revealing, quantized, or randomized policies. Notably, the first arises in statistical decision theory and classical communication theoretic problems involving a fully aligned sensor and receiver, the second arises in Nash theoretic simultaneous signaling games, and the last may appear in Stackelberg type (leader-follower) Bayesian signaling games. In this paper, we investigate the Bayesian persuasion problem involving a receiver that tries to estimate the source. We show that for certain payoff structures, the equilibrium solution is such that a source observation is mapped to distinct messages with nonzero probabilities. More specifically, we completely characterize conditions under which the sender requires randomization for the Bayesian persuasion problem involving general sources with finite cardinality. In particular, regardless of whether the equilibrium solution under a deterministic policy restriction is fully revealing, quantized or noninformative, there exists a randomized sender policy that improves the sender's payoff under certain conditions characterized in the paper. Moreover, we provide an algorithmic procedure to obtain the Bayesian persuasion solution, where the algorithm compares the payoffs with finitely many posterior probability combinations. We also consider fully aligned and completely misaligned payoff structures, where the solutions respectively involve a fully revealing sender and a noninformative sender. Then, we unify these results by proving that if the sender's expected payoff with respect to posterior distributions is continuous, then the equilibrium solution involves either a fully revealing sender or a noninformative sender.
\end{abstract}


\section{Introduction and Related Works}
\label{sec:intro}

In classical control, communication and information theoretic problems, a feature that is classically common is that decision makers share a common goal. Such systems are referred to as {\it team theoretic} since the decision makers in the system form a team to optimize a common objective. In a team theoretic problem, a classical result (also related to Blackwell \cite{Blackwell1953}) states that sharing more information with other decision makers does not harm any decision maker in the system. In other words, a decision maker in a team theoretic system does not wish to withhold or distort information with the purpose of achieving its goal.

On the other hand, there may be decision makers who have misaligned objectives in a system, requiring a {\it game theoretic} analysis. For instance, in a decentralized control system, malicious agents may wish to distort the control performance while legitimate agents desire to optimize the control objective \cite{MaliciousSurvey2022,MaliciousTac2017,MaliciousTcns2022}. In such a communication system, a sender may wish to convey certain information to a receiver while protecting certain sensitive information from the receiver, where the sensitive information is correlated with the actual information \cite{AsoDiaAlaLin2019,IssWagKam2020}. In such game theoretic settings, a decision maker needs to strategically design what information to share with others by taking its own goal as well as other's goals into account. Unlike team theoretic settings, sharing more information may hurt a decision maker or even all the decision makers in the system \cite{BasGosScaZam2003,hogeboom2021continuity}. Therefore, a decision maker in a game theoretic setting may wish to hide or distort information to achieve a certain goal.

Such a setup between a sender and a receiver with misaligned objectives is referred to as a signaling game. In this paper, we focus on a specific signaling game problem of Stackelberg type, commonly referred to as Bayesian persuasion initiated by Kamenica and Gentzkow in \cite{KamGen2011}. In Bayesian persuasion, a sender wishes to steer the opinion of a receiver by designing its information revelation policy. In this paper, we aim to investigate the structure of the sender's policy for this fundamental problem. In particular, the problem of when and how the sender hides or distorts information is of significance from an information theoretic perspective with applications in a wide variety of fields.

In this field, a seminal work by Crawford and Sobel \cite{CraSob1982} introduces a signaling game problem where a sender and a receiver communicate in a Nash theoretic simultaneous setting. It is shown in \cite{CraSob1982} that at a Nash equilibrium, the sender needs to hide information from the receiver by mapping continuous source observations to discrete messages. On the other hand, Kamenica and Gentzkow investigate a Stackelberg type signaling game where again a sender conveys information to a receiver \cite{KamGen2011}. In this setting, the sender (leader) and then the receiver (follower) choose their strategies where the receiver acts given the strategy of the sender. Although the sender is transparent in terms of its choice of policy, the sender may convince the receiver to act according to its preference to a certain extent by strategically designing its information revelation policy. In this paper, we investigate conditions under which the sender reveals, hides or distorts information transmitted to the receiver for the Bayesian persuasion setting.

Signaling games problems are investigated in numerous studies in communication, control and economics literature \cite{AloCam2016,DugXu2016,Dughmi2017,KamXia2025,Tamura2018,AkyLanBas2017,SarYukGez2017,KazSarGezLinYuk2022,KazGezYuk2022,TsaTsa2021,TreTom2016,TreTom2021,BouTre2022,DeoKul2024Arxiv,VorKul2024,Ichihashi2019,Wu2023}. The work in \cite{AloCam2016} studies a Bayesian persuasion setting with the sender and the receiver having different prior beliefs regarding the source, and provide conditions under which the sender improves its payoff by communicating with the receiver. The work in \cite{DugXu2016} investigates computational complexity for obtaining the solution of the Bayesian persuasion problem. In particular, the problem is expressed as a linear program, and it is shown that the solution can be obtained in polynomial time, considering independent and identically distributed sources. In addition, the work in \cite{Dughmi2017} reviews computational methods for obtaining the Bayesian persuasion problem considering various settings such as with multiple agents. In \cite{KamXia2025}, the authors investigate whether commitment is beneficial for the sender by contrasting the simultaneous move setup with the Bayesian persuasion setup, and connect the notions of commitment at the sender and randomized policies. The Bayesian persuasion problem with Gaussian sources are investigated in \cite{Tamura2018,AkyLanBas2017}. In particular, \cite{Tamura2018} shows that the Bayesian persuasion solution for Gaussian sources and quadratic cost structures involves linear sender policies. In \cite{AkyLanBas2017}, information theoretic limits are investigated for a strategic communication setup with a biased sender where the source and the bias are Gaussian. The work in \cite{SarYukGez2017} considers continuous sources and investigates the Stackelberg equilibrium (Bayesian persuasion) as well as the Nash equilibrium in a signaling game setup with a deterministically biased sender. The work in \cite{KazSarGezLinYuk2022} analyzes the signaling game setup introduced in the seminal paper by Crawford and Sobel \cite{CraSob1982} and investigates the number of messages that can be sent at a Nash equilibrium with a continuous source. In \cite{KazGezYuk2022}, a privacy perspective is investigated for a signaling game setup considering Gaussian sources where the sender partially hides information due to privacy concerns. 

In addition, \cite{TsaTsa2021} considers a Bayesian persuasion setting with a discrete noisy symmetric channel between the sender and the receiver. The authors investigate informativeness in the Blackwell sense and show that a noisier symmetric channel leads to a lower utility for the sender. In this context, regularity properties of equilibria when information structures are perturbed are studied in \cite{hogeboom2021continuity,YukBas2024}. In \cite{TreTom2016}, the authors consider a Bayesian persuasion setting with a binary discrete source and a symmetric channel between the sender and the receiver and investigates equilibrium solution. In \cite{TreTom2021}, the authors incorporate decoder side information to a Bayesian persuasion setup with a noisy channel and investigates information theoretic limits of such a strategic setup (see also \cite{BouTre2021} where the focus is on utilities). The work in \cite{BouTre2022} considers a signaling game setup with two decoders where the decoders successively decode received messages and analyzes information theoretic limits of such a communication setup. In \cite{DeoKul2024Arxiv}, the authors investigate informativeness in the sense of minimum information that a receiver can recover by considering the worst case expected sender utility and the notion of $\epsilon$-Stackelberg equilibrium for the Bayesian persuasion problem with a uniformly distributed source. In \cite{VorKul2024}, a leader-follower type interaction is considered where the roles of the leader and the follower are flipped compared to the Bayesian persuasion setting. In such strategic interactions, the receiver in a sense chooses its policy to convince the sender to reveal as much information as possible. In \cite{Ichihashi2019}, a designer is incorporated into a sender and a receiver interaction where the designer limits the information (in Blackwell sense) that the sender can transmit. The work in \cite{Wu2023} investigates a Bayesian persuasion setup with a sequential interaction and obtains the equilibrium solution via recursively applying the concavification method described in \cite{KamGen2011}. In addition, a related line of work considers improving system performance by applying randomization at the sender or at the receiver, where the metrics involve probability of error of the source symbols at the receiver \cite{GokGezAri2010,DulGez2012}.  

Further related to our work, \cite{DavChrMal2022} is motivated by the spread of misinformation on social media, and a game theoretic approach is proposed where the government incentives social media platforms to reveal truthful information. In \cite{LuoRoz2025}, the sender is allowed to lie by transmitting a different message than the message generated by the committed encoding policy in the Bayesian persuasion setup, and the receiver has a mechanism to detect such misleading information with a certain probability. In \cite{EdeMin2022}, a binary source Bayesian persuasion setup with the sender acting according to the committed policy is considered where a message is interpreted as a lie if it does not match with the true source value and the receiver can probabilistically detect such a lie.

\section{Problem Formulation}
\label{sec:prob}

We consider the Bayesian persuasion problem \cite{KamGen2011} where a sender communicates with a receiver. The sender observes a discrete source $X\in\mathcal{X}$ where $X$ is an $\mathcal{X}$-valued random variable, and $\mathcal{X}$ is a finite set denoting the set of possible source observations. The sender maps its source observation to a message $M\in \mathcal{M}$ by employing an encoding policy $\gamma^s(\cdot)$ where $\mathcal{M}$ denotes the set of possible messages. In the Bayesian persuasion framework, the encoding policy is announced to the receiver. Then, the receiver takes an action $Y\in\mathcal{Y}$ given its message observation by using a decoding policy $\gamma^r(\cdot)$ where $\mathcal{Y}$ denotes the set of possible receiver actions. The communication setting is illustrated in Fig.~\ref{fig:comm}.

\begin{figure}
\begin{minipage}[b]{1.0\linewidth}
\centering
\centerline{\includegraphics[width=.8\textwidth]{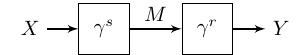}}
\end{minipage}
\caption{Communication setting.}
\label{fig:comm}
\end{figure}

The communication scenario corresponds to the classical source coding problem with a discrete source observation. However, unlike the classical source coding setting, the sender and the receiver do not share a common objective. In contrast, the sender wishes to have the receiver act in a certain fashion by designing the information that it provides to the receiver. We model the misaligned objectives of the sender and the receiver by expected utilities. In particular, the sender wishes to maximize 
\begin{align}
U^s(\gamma^s,\gamma^r) = \mathbb{E}[u^s(X,Y)],
\end{align}
whereas the receiver wishes to maximize
\begin{align}
U^r(\gamma^s,\gamma^r) = \mathbb{E}[u^r(X,Y)],
\end{align}
where $u^s(x,y)$ and $u^r(x,y)$ respectively denote the sender's utility and the receiver's utility when the source observation is $x\in\mathcal{X}$ and the receiver action is $y\in\mathcal{Y}$. The receiver chooses its policy given the announced encoding policy and knows that the sender is committed to this announced encoding policy. In other words, the sender does not deviate from its announced policy. This corresponds to a Stackelberg equilibrium notion \cite{BasOls1999} with the sender being the leader and the receiver being the follower. In other words, the aim in the Bayesian persuasion is to obtain the Stackelberg equilibrium solution where the sender chooses its policy first.

We now formally describe the Stackelberg equilibrium. Let $\gamma^s\in\Gamma^s$ and $\gamma^r\in\Gamma^r$, where $\Gamma^s$ and $\Gamma^r$ respectively denote the set of possible encoding and decoding policies. A set of policies $\{\gamma^{s,*},\gamma^{r,*}\}$ is referred to as a Stackelberg equilibrium if 
\begin{align}
U^s(\gamma^{s,*},\gamma^{r,*}(\gamma^{s,*}))
\geq 
U^s(\gamma^{s},\gamma^{r,*}(\gamma^{s}))
\label{eq:Stackelberg_def_1}
\end{align}
is satisfied for all $\gamma^{s}\in \Gamma^s$, where $\gamma^{r,*}(\gamma^{s})$ is such that
\begin{align}
U^r(\gamma^{s},\gamma^{r,*}(\gamma^{s}))
\geq 
U^r(\gamma^{s},\gamma^{r}(\gamma^{s}))
\label{eq:Stackelberg_def_2}
\end{align}
holds for all $\gamma^{r}\in \Gamma^{r}$. The notation $\gamma^{r}(\gamma^{s})$ is used to emphasize that the receiver chooses its policy given the sender's policy announcement. Here, we assume that when the receiver is indifferent between several actions, it chooses the action that the sender prefers \cite{KamGen2011}. A motivation for this assumption is that the sender may employ a slightly perturbed policy in order to induce the action that it prefers. In the following, we formalize this assumption.

\begin{assumption}\label{assumption:tiebreaking}
[Sender-Preferred Tie Breaking]
Let $\Gamma^{r,*}(\gamma^s)$ denote the set of receiver policies that attains the maximum in \eqref{eq:Stackelberg_def_2} for a given sender policy $\gamma^s$. If $\Gamma^{r,*}(\gamma^s)$ contains more than one policy, then the receiver employs a policy $\gamma^{r,*}\in\Gamma^{r,*}(\gamma^s)$ that satisfies
\begin{align}
U^s(\gamma^s,\gamma^{r,*}(\gamma^s)) 
\geq 
U^s(\gamma^s,\gamma^{r}(\gamma^s))
\end{align}
for all $\gamma^{r}\in\Gamma^{r,*}(\gamma^s)$. 
\end{assumption}

In this paper, our aim is to investigate the structure of encoding policies in equilibrium. We make the following distinction between encoding policies with respect to their structure.

\begin{definition}\label{def:deterministic}
\begin{itemize}
\item[(i)] [Truthful Information Hiding] A non-injective encoding policy is referred to as a \emph{lossy quantized encoding policy} if under $\gamma^s$, $P^{\gamma^s}(M=m\vert X=x)\in \{0,1\}$ for all $m\in\mathcal{M}$ and $x\in\mathcal{X}$, and there exist $x_1,x_2\in\mathcal{X}$ and $m\in\mathcal{M}$ such that $P^{\gamma^s}(M=m\vert X=x_1)=1$ and $P^{\gamma^s}(M=m\vert X=x_2)=1$. 
\item[(ii)] [Fully Revealing Policy] An encoding policy is referred to as a \emph{fully revealing encoding policy} if all source observations are mapped to distinct messages. 
\item[(iii)] [Noninformative Encoding Policy] An encoding policy is referred to as a \emph{noninformative encoding policy} if the encoded message is independent of the source random variable.
\end{itemize}
\end{definition}

\begin{remark}
A noninformative encoding policy may be viewed as a lossy quantized encoding policy where all source observations are mapped to the same message. 
\end{remark}

Definition~\ref{def:deterministic} categorizes encoding policies that essentially do not apply randomization. In the following, we formally make a distinction between encoding policies that use randomization or not. 

\begin{definition}\label{def:randomized}
[Misinformation Policy] An encoding policy $\gamma^s$ is referred to as a \emph{randomized encoding policy} if there exist $x_1,x_2\in\mathcal{X}$ and $m_1,m_2\in\mathcal{M}$ such that $P^{\gamma^s}(M=m_1\vert X=x_1)>0$, $P^{\gamma^s}(M=m_2\vert X=x_1)>0$ and $P^{\gamma^s}(M=m_1\vert X=x_2)>0$ hold. Such an encoding policy in a sense leads to \emph{misinformation} since the sender tries to confuse the receiver by mixing distinct source observations with certain nonzero probabilities.
\end{definition}

\begin{remark}
Definition~\ref{def:randomized} categorizes a policy as a randomized policy if an equivalent deterministic policy does not exist. Namely, for a policy $\gamma^s$ with $P^{\gamma^s}(M=m_1\vert X=x_1)>0$, $P^{\gamma^s}(M=m_2\vert X=x_1)>0$ and $P^{\gamma^s}(M=m_1\vert X=x)=P^{\gamma^s}(M=m_2\vert X=x)=0$ for all $x\neq x_1$, we can obtain an equivalent policy ${\gamma^{s}}'$ with $P^{{\gamma^{s}}'}(M=m'\vert X=x_1) =P^{\gamma^s}(M=m_1\vert X=x_1)+P^{\gamma^s}(M=m_2\vert X=x_1)$ where the rest of the mappings for $\gamma^{s}$ and ${\gamma^{s}}'$ are the same. This procedure can be applied repeatedly to obtain an equivalent deterministic policy if the initial policy is not categorized as a randomized policy in Definition~\ref{def:randomized}.
\end{remark}

A related work \cite{Sobel2020} formalizes the notion of a lie where a message transmission is interpreted as a source (or sources) revelation and a lie means that a message observation is not consistent with the corresponding interpretation of the message. In view of \cite{Sobel2020}, a randomized signaling may be interpreted as \emph{lying} since the sender mixes multiple source realizations to a single message in this case.

While in a different operational context, such lossy quantized encoding policies arise in equilibrium solutions for signaling games (in the Nash sense) as studied in the seminal paper by Crawford and Sobel \cite{CraSob1982}. While the sender must apply a lossy quantized encoding (or an equivalent policy) to continuous source observations at a Nash equilibrium in \cite{CraSob1982}, the Bayesian persuasion solution may involve lossy quantized encoding policies applied to discrete source observations, depending on the objectives. In addition, while randomization is irrelevant in \cite{CraSob1982}, the Bayesian persuasion solution may require randomization at the sender depending on the payoffs.

Whether in an equilibrium, coding/signaling policies are fully revealing, quantized, or randomized has significant implications in applications: The first arises in team theoretic classical communication problems as well is in statistical decision theory \cite{Blackwell1953}, the second arises in Nash theoretic simultaneous signaling games \cite{CraSob1982}, and the last, as we will study further, may appear in Stackelberg type (leader-follower) Bayesian signaling games which have recently been given much attention in part thanks to the highly influential study by Kamenica and Gentzkow \cite{KamGen2011}. In addition to different quantitative and mathematical properties of such equilibrium solutions; at the qualitative level, signaling behavior of the sender has important implications in applications: The fully revealing case is one where there is complete transparency, the randomized signaling leads to a general representation of what is typically referred to as misinformation which is a highly prevalent phenomenon of our digital age (see e.g., \cite{DavChrMal2022,BodSchWatFun2024,HosMlaCheGid2024}), and the quantized case is where there is information hiding but no misinformation. Therefore, it is an important problem to identify conditions that lead to such distinct behavior in signaling games.

{\bf Contributions.}
\begin{itemize}
\item In Theorems~\ref{thm:fully_revealing}-\ref{thm:quantized}, we derive conditions under which the Bayesian persuasion solution involves a randomized policy at the sender, where we consider the case when the receiver wishes to estimate the source. In particular, depending on the structure of the equilibrium solution under a deterministic policy restriction, we write explicit conditions under which randomization improves the sender's payoff. In other words, for certain payoff structures, the equilibrium solution is such that a source observation is mapped to distinct messages with nonzero probabilities (and thus randomized encoding is superior to deterministic coding; that is, misinformation is optimal for the sender).
\item In Theorem~\ref{thm:algorithm}, we derive an algorithmic procedure to obtain the Bayesian persuasion solution, where the procedure involves the computation of expected payoffs for finitely many posterior probability combinations. While the algorithm considers a receiver that tries to extract the source, the result generalizes to arbitrary payoff structures.
\item In Theorem~\ref{thm:continuous}, we show that if the sender's expected payoff is continuous with respect to posterior distributions, then the equilibrium solution involves either a fully revealing sender or a noninformative sender. For such payoff structures, a slight perturbation in a sender policy leads to a slight change in the sender's expected payoff when the receiver employs its best response. While this result considers a receiver that tries to estimate the source, it generalizes to more general receiver payoff structures stated in the manuscript. In addition, in Theorem~\ref{thm:zerosum}, we consider arbitrary zero-sum payoff structures and prove that the equilibrium solution induces noninformative signaling in this case. 
\end{itemize}

\section{Preliminaries}
\label{sec:game}

In this section, we first review the solution approach considered in the seminal paper by Kamenica and Gentzkow \cite{KamGen2011}. In particular, the equilibrium solution is posed as an optimization problem involving posterior distributions. For a message observation $m\in\mathcal{M}$, the receiver computes the posterior distribution of the source to determine its action. In particular, for a given computed posterior distribution $Q(\cdot)$, the receiver's action is given by
\begin{align}
y^*(Q) = \argmax_{y\in\mathcal{Y}} \sum_{x\in\mathcal{X}} u^r(x,y)Q(x)
\label{eq:ystar}
\end{align}
where the receiver chooses the action that gives a higher payoff for the sender in case multiple receiver actions yield the optimal value in \eqref{eq:ystar} as in \cite{KamGen2011}. For a particular posterior distribution $Q(\cdot)$, the contribution to the total payoff of the sender can be written as
\begin{align}
u^{s,*}(Q) = \sum_{x\in\mathcal{X}} u^s(x,y^*(Q))Q(x).
\label{eq:usstar}
\end{align}
Let $Q_m(\cdot)$ denote a posterior distribution computed by the receiver for a message observation $m\in\mathcal{M}$. Then, for a given encoding policy $\gamma^s$, the sender's expected payoff can be written as 
\begin{align}
U^{s,*}(\gamma^s)=
\sum_{m\in\mathcal{M}} 
u^{s,*} (Q_m) P(Q_m)
\label{eq:UeStar}
\end{align}
where $P(Q_m)$ denotes the probability of the posterior distribution $Q_m$ (i.e., probability that message $m$ is sent). The sender wishes to maximize \eqref{eq:UeStar} by choosing an encoding policy that induces a set of posteriors with certain probabilities. In order to verify that a set of posterior distributions is induced by a valid encoding policy, one can check if these posteriors and their probabilities are consistent with the prior distribution of the source. In particular, these posterior distributions and their probabilities must satisfy a \emph{Bayes plausibility} condition expressed in the following:
\begin{align}
\pi_x = 
\sum_{m\in\mathcal{M}} Q_m(x) P(Q_m),
\quad
\text{for all }x\in\mathcal{X},
\label{eq:bayesplausibility}
\end{align}
where $\pi_x\triangleq P(X=x)$. Then, we can express the sender's optimization problem as 
\begin{align}
&\max_{
\{Q_m\}_{m=0}^{|\mathcal{M}|-1}, \{ P(Q_m)\}_{m=0}^{|\mathcal{M}|-1}
}
\sum_{m\in\mathcal{M}} 
u^{s,*} (Q_m) P(Q_m) \nonumber\\
&\subjectto\,
\pi_x = 
\sum_{m\in\mathcal{M}} Q_m(x) P(Q_m)\;
\text{for all }x\in\mathcal{X}.
\label{eq:optim_problem}
\end{align}
While the optimization problem in \eqref{eq:optim_problem} is a very insightful formulation, it is nonetheless not easy to solve since the space of attainable posteriors is not an easy set to characterize. In particular, the formulation in \eqref{eq:optim_problem} performs an optimization on the space of probability distributions on posterior distributions, yielding an extremely large set. In addition, a Bayesian plausibility constraint is imposed, making the problem even more challenging. 

\begin{lemma}\label{lem:nonconvex_posteriors}
The set of attainable posterior combinations in \eqref{eq:optim_problem} is not convex. 
\end{lemma}
\begin{proof}
See Appendix~\ref{proof:nonconvex_posteriors}.
\end{proof}

The Bayesian plausibility condition in \eqref{eq:bayesplausibility} states that a posterior distribution combination is feasible when the prior distribution of the source is a convex combination of this posterior distribution combination. Then, if we have a convex or concave $u^{s,*}(Q)$, we obtain the following result by applying the convexity or concavity property to the optimization problem in \eqref{eq:optim_problem} together with the Bayesian plausibility property. 

\begin{theorem}[{Kamenica and Gentzkow \cite{KamGen2011}}]\label{thm:KamGen}
If $u^{s,*}(Q)$ is concave, the sender's payoff does not improve by revealing information about the source for any prior distribution. If $u^{s,*}(Q)$ is convex and not concave, the optimal sender reveals at least partial information about the source for every prior distribution.
\end{theorem}

While the result in Theorem~\ref{thm:KamGen} is a useful observation, it imposes a continuity requirement on $u^{s,*}(\cdot)$ dictated by the convexity or concavity property. 

\section{Properties of Equilibrium Behavior}

\subsection{Completely Aligned and Misaligned Payoffs}

In this section, we investigate two extremes of completely aligned and completely misaligned payoffs. We first revisit the classical setting where the sender and the receiver have aligned objectives. This setting is referred to as team theoretic since the sender and the receiver form a team to optimize a common objective \cite{YukBas2024}. In this case, a classical result by Blackwell \cite{Blackwell1953} applied to our setting with aligned objectives states that hiding information does not benefit the decision makers. In the following, we state this result and follow the proof given in \cite[Theorem~7.3.2]{YukBas2024}.

\begin{theorem}[{Blackwell \cite{Blackwell1953}}]\label{thm:blackwell}
Suppose that $u^s(x,y)=u^r(x,y) $ holds for all $x\in\mathcal{X}$ and $y\in\mathcal{Y}$. Then, the Stackelberg equilibrium solution involves a fully revealing encoding policy.
\end{theorem}

\begin{proof}
See Appendix~\ref{proof:blackwell}.
\end{proof}

\begin{remark}
The proof of Theorem~\ref{thm:blackwell} relies on a convexity argument and Jensen's inequality. In addition, the result of Theorem~\ref{thm:blackwell} is still valid when there is a noisy channel between the sender and the receiver.
\end{remark}

Theorem~\ref{thm:blackwell} addresses the case when the payoffs are completely aligned and shows that the equilibrium solution is fully revealing in this case. Next, we consider the other extreme case when the payoffs are completely misaligned by taking a zero-sum payoff structure. For zero-sum games, it has been shown recently that more information does not hurt the party receiving it \cite{pkeski2008comparison, HBY2020arXiv,hogeboom2021continuity}. Consistent with this observation, in the following, we show that the equilibrium solution becomes noninformative under such completely misaligned payoff structures.

\begin{theorem}\label{thm:zerosum}
Suppose that $u^s(x,y)=-u^r(x,y)$ holds for all $x\in\mathcal{X}$ and $y\in\mathcal{Y}$. Then, the equilibrium solution involves a noninformative encoding policy.
\end{theorem}

\begin{proof}
See Appendix~\ref{proof:zerosum}.
\end{proof}

In view of Theorems \ref{thm:KamGen}-\ref{thm:zerosum}, in this paper, we provide further results on when an equilibrium solution may be fully revealing, quantized, or randomized in a Stackelberg type Bayesian signaling game.

\subsection{Continuous Expected Payoff for the Sender}

In this section, we consider a receiver that tries to estimate the source with its action. More specifically, we let $\mathcal{X}=\mathcal{Y}=\{0,1,\dots,M-1\}$, $u^r(x,y)=1$ for all $x=y$, and $u^r(x,y)=0$ for all $x\neq y$, where $M$ denotes the cardinality of the source/action space. We note that it is possible to relabel the source and action values. This allows us to take $\pi_0\geq \pi_1\geq\cdots\geq \pi_{M-1}$ without loss of generality for notational convenience. In addition, if $\pi_i=\pi_j$ is satisfied with $i<j$, we let $u^s(i,i)+u^s(j,i) \geq u^s(i,j)+u^s(j,j)$ so that the receiver takes the action $y=i$ when indifferent between $y=i$ and $y=j$.

In Theorem~\ref{thm:blackwell} and Theorem~\ref{thm:zerosum}, it is shown that the equilibrium solution involves a fully revealing sender and a noninformative sender in a team setup and a zero-sum setup, respectively. In these cases, convex analytical arguments are employed to prove the results where the payoffs are arbitrary. We note that while convexity/concavity implies continuity, the converse is not necessarily true. In the following, we make a continuity assumption and present a result that is consistent with Theorems~\ref{thm:blackwell}-\ref{thm:zerosum}. Interestingly, it is observed that for the Bayesian persuasion problem, the considered continuity property also leads to either fully revealing or noninformative signaling as the equilibrium solution for the considered receiver payoff structure. More specifically, we employ a linear algebraic argument and show that when the function $u^{s,*}(\cdot)$ is continuous, neither randomization nor quantization with partial information revelation improves the sender's payoff. To prove this result, we need an auxiliary result that connects the continuity of $u^{s,*}(\cdot)$ to a condition involving the payoffs.

\begin{lemma}\label{lem:continuity}
We have that $u^{s,*}(\cdot)$ is continuous if and only if the following holds. For all posteriors $Q$ satisfying 
\begin{align}
\sum_{x\in\mathcal{X}}u^r(x,y)Q(x)= \sum_{x\in\mathcal{X}}u^r(x,y')Q(x)>\sum_{x\in\mathcal{X}}u^r(x,y'')Q(x)
\label{eq:lem_continuity_eq1}
\end{align}
for all $y,y'\in A$ and $y''\notin A$ with $A\subseteq \mathcal{Y}$, we have that
\begin{align}
\sum_{x\in\mathcal{X}}u^s(x,y)Q(x)= \sum_{x\in\mathcal{X}}u^s(x,y')Q(x) 
\label{eq:lem_continuity_eq2}
\end{align}
for all $y,y'\in A$.
\end{lemma}
\begin{proof}
See Appendix~\ref{proof:lem_continuity}.
\end{proof}

Lemma~\ref{lem:continuity} demonstrates that for all posteriors inducing equal receiver preferences for certain actions, the sender must be indifferent between those actions in order to make $u^{s,*}(\cdot)$ continuous. An interpretation is that since the receiver chooses an action that gives a higher payoff for the sender for these posteriors due to Assumption~\ref{assumption:tiebreaking}, the conditions in Lemma~\ref{lem:continuity} must be satisfied to avoid discontinuities at these posteriors. In the following, we employ this equivalent condition to prove the desired result.

\begin{theorem}\label{thm:continuous}
If $u^{s,*}(\cdot)$ is continuous, then the equilibrium solution involves either a fully revealing or a noninformative encoding policy regardless of the prior distribution.
\end{theorem}
\begin{proof}
See Appendix~\ref{proof:continuous}.
\end{proof}

We note that while Theorem~\ref{thm:blackwell} and \ref{thm:zerosum} are consistent, as special cases, with Theorem 4; the result is strictly more general than these two extreme cases.

\begin{remark}
In the proof of Theorem~\ref{thm:continuous}, we consider a receiver that wishes to estimate the source. Nevertheless, the result is also valid for more general receiver payoff structures. In particular, if there exists a posterior with all nonzero terms under which the receiver is indifferent between all the actions, then Theorem~\ref{thm:continuous} can be applied. Formally, the result of Theorem~\ref{thm:continuous} is still valid if there exists a posterior $Q$ with $Q(x)>0$ for all $x\in\mathcal{X}$ and 
\begin{align}
\sum_{x\in\mathcal{X}}u^r(x,y)Q(x)=\sum_{x\in\mathcal{X}} u^r(x,y')Q(x)
\end{align}
for all $y,y'\in\mathcal{Y}$.
\end{remark}

\begin{remark}
In the proof of Theorem~\ref{thm:continuous}, we apply a linear algebraic argument to prove the result. In particular, this argument reveals that  $u^{s,*}(\cdot)$ is continuous with respect to posterior distributions if and only if the sender's payoff $u^s(\cdot,\cdot)$ is essentially unique up to a multiplicative factor. This multiplicative coefficient determines whether the solution involves a noninformative or fully revealing sender. 
\end{remark}

\subsection{An Algorithmic Procedure for Obtaining Equilibrium Solution}

In the remainder of this section, we investigate conditions under which the equilibrium solution involves a fully revealing, quantized, or randomized sender. Towards that goal, we first present an algorithm to obtain the equilibrium solution. This algorithm computes and compares the expected payoffs under a number of feasible posterior probability combinations. This algorithmic formulation facilitates the investigation of the equilibrium structure. 

\begin{theorem}\label{thm:algorithm}
Let $\mathcal{X}=\mathcal{Y}=\{0,1,\dots,M-1\}$, $u^r(x,y)=1$ for all $x=y$, and $u^r(x,y)=0$ for all $x\neq y$. Then, the Bayesian persuasion solution can be obtained by following the procedure in Algorithm~\ref{alg:BayPers}.
\end{theorem}

\begin{proof}
See Appendix~\ref{proof:algorithm}.
\end{proof}

\begin{algorithm}
\caption{An algorithm to obtain the Bayesian persuasion solution.}
\label{alg:BayPers}
\begin{algorithmic}
\STATE $U^{s,*} = -\infty$
\STATE $\mathcal{Q} = \{[1,0,\dots,0]^T,\dots,[0,\dots,0,1]^T\} \triangleq \{e_0,\dots,e_{M-1}\}$
\FOR{each $\mathcal{A} \subseteq \mathcal{X}$ with $|\mathcal{A}|=K\geq 2$} 
\STATE Let $Q(x) = 1/K$ for all $x\in\mathcal{A}$ and $Q(x)=0$ for all $x\notin\mathcal{A}$.
\STATE $\mathcal{Q} \gets \mathcal{Q} \cup \{Q\}$
\ENDFOR
\FOR{each linearly independent $\{Q_i\}_{i=0}^{M-1}$ with $Q_i\in \mathcal{Q}$}
\STATE Check if one can find $\{\alpha_i\}_{i=0}^{M-1}$ satisfying $0\leq \alpha_i\leq 1$, $\sum_{i=0}^{M-1} \alpha_i=1$ and
\begin{align}
\pi_x = \sum_{i=0}^{M-1} \alpha_i Q_i(x) \quad\text{for all } x\in\mathcal{X}.
\label{eq:algo_eq3}
\end{align}
\STATE For feasible $\{Q_i\}_{i=0}^{M-1}$, compute
\begin{align}
U^{s}=
\sum_{i=0}^{M-1} 
\alpha_i u^{s,*} (Q_i).
\label{eq:algo_eq4}
\end{align}
\IF{$U^{s,*}<U^{s}$}
\STATE $U^{s,*} \gets U^{s}$
\ENDIF
\ENDFOR
\end{algorithmic}
\end{algorithm}

\begin{remark}
In the first step of Algorithm~\ref{alg:BayPers}, we form a posterior set $\mathcal{Q}$, which contains all the posterior distributions that may appear in an equilibrium solution. In other words, Theorem~\ref{thm:algorithm} shows that there is no need to consider any other posterior distribution while obtaining the equilibrium solution. In the remainder of the manuscript, the set $\mathcal{Q}$ is referred to as the {\it critical posterior set}. 
\end{remark}

\begin{remark}
The result of Theorem~\ref{thm:algorithm} generalizes to arbitrary receiver payoff structures. In this case, the critical posterior set $\mathcal{Q}$ needs to be modified. In particular, for all $A\subseteq\mathcal{X}$, it is required to find the posterior(s) $Q$ under which the receiver is indifferent between the actions corresponding to the sources in $A$ with $Q(x)=0$ for all $x\notin A$. 
\end{remark}

\begin{remark}
In addition to $e_i$'s, the critical posterior set $\mathcal{Q}$ contains posteriors under which the receiver is indifferent between multiple actions. By inducing such posteriors, the sender essentially confuses the receiver among multiple actions. This can be interpreted as misinformation since the sender distorts its source observations with the purpose of inducing a mismatched action for certain source realizations. In addition, for a given action combination $\mathcal{A}\subset \mathcal{Y}$, there exists a family of posteriors inducing equal receiver preferences for these actions. Nevertheless, we show that it is sufficient to consider only the unique posterior distribution under which the receiver is certain that the source satisfies $X\in \mathcal{A}$ when the corresponding message is received. In other words, the sender does not benefit by confusing the receiver with an additional source realization unless its posterior probability is equal to the rest of the sources in $\mathcal{A}$.
\end{remark}

Theorem~\ref{thm:algorithm} describes a procedure to obtain the equilibrium solution where one first checks feasibility for a number of posterior combinations and then computes the payoffs for feasible ones. It is seen that the number of posterior combinations may become large when the cardinality of the source/action set is large. However, by taking the Bayesian plausibility condition into account during the selection of posterior combinations in Algorithm~\ref{alg:BayPers}, a reduced set can be obtained. 

\begin{remark}\label{rem:reduction}
We know that the actions $y=0$, $y=1$, $\dots$, $y=M-1$ are sorted in descending order of preference from the receiver's perspective under the prior distribution. This implies that a feasible posterior combination must include a posterior under which the receiver prefers $y=i$ over $y=j$ for all $i$ and $j$ satisfying $0\leq i< j \leq M-1$. 
\end{remark}

\begin{remark}
One may need to consider a large number of posterior probability combinations even after applying the reduction in Remark~\ref{rem:reduction} when the cardinality of the source/action set is large. However, if one restricts attention to deterministic policies, the size of the feasible posterior combination set becomes significantly small. In other words, the difficulty of solving the Bayesian persuasion problem mainly arises when we allow randomized policies at the sender. This motivates us from an analytical perspective to investigate conditions under which the equilibrium solution involves a randomized sender.  
\end{remark}

\subsection{Explicit Conditions for the Optimality of Fully Revealing, Randomized, Quantized, or Noninformative Signaling}

In this section, we also consider a receiver that wishes to estimate the source and write explicit conditions under which the equilibrium solution involves a fully revealing, noninformative, quantized, or randomized sender. Towards that goal, we first derive conditions ensuring that the equilibrium solution involves a fully revealing, noninformative, or quantized sender under a deterministic policy restriction. Then, under these separate conditions, we investigate whether a randomized policy improves the sender's payoff. In the following, we first investigate the optimality of a fully revealing policy.

\begin{lemma}\label{lem:fully_revealing_lemma}
The optimal sender under a deterministic policy restriction is fully revealing if and only if the following inequalities hold:
\begin{align}
u^s(i,i)\geq u^s(i,j) 
\label{eq:lem_fully_reveal_eq1}
\end{align}
for all $i,j\in\{0,1,\dots,M-1\}$ satisfying $\pi_j>\pi_i$, and 
\begin{align}
u^s(i,i) \geq u^s(i,j)\text{ and }u^s(j,j) \geq u^s(j,i)
\label{eq:lem_fully_reveal_eq2}
\end{align}
for all $i,j\in\{0,1,\dots,M-1\}$ satisfying $\pi_i=\pi_j$.
\end{lemma}

\begin{proof}
See Appendix~\ref{proof:fully_revealing_lemma}
\end{proof}

\begin{remark}
With a fully revealing encoding policy, the receiver's action is equal to the source observation at the sender. For the optimality of a fully revealing policy under a deterministic policy restriction, it is required that the sender does not wish to induce a different action $y=j$ for a source observation of $x=i$ only when $\pi_j\geq \pi_i$ holds. On the other hand, for $\pi_j< \pi_i$, such conditions do not appear in Lemma~\ref{lem:fully_revealing_lemma} since the sender cannot convince the receiver to take the action $y=j$ under a source observation of $x=i$ with a deterministic encoding policy in this case.
\end{remark}

Next, we assume that the conditions of Lemma~\ref{lem:fully_revealing_lemma} are satisfied and show that the sender's payoff improves via randomization at the sender under certain conditions. That is, misinformation via mixing source observations outperforms truthful source revelation under certain conditions.

\begin{theorem}\label{thm:fully_revealing}
Suppose that the optimal sender among deterministic policies is fully revealing. Then, the equilibrium solution involves a randomized encoder if and only if applying pairwise randomization improves the sender's payoff for at least one pair of source values. 
\end{theorem}

\begin{proof}
See Appendix~\ref{proof:fully_revealing_thm}.
\end{proof}

\begin{remark}
Theorem~\ref{thm:fully_revealing} reveals that when the optimal encoder among deterministic policies is fully revealing, it is enough to apply Theorem~\ref{thm:algorithm} for each source pairs to see whether randomization is needed for the equilibrium solution. In other words, for each source pair, the sender does not gain by confusing the receiver among these sources. However, when the equilibrium solution involves a randomized sender, the solution may involve a sender that maps three or more source realizations to the same message with certain nonzero probabilities.     
\end{remark}

\begin{remark}
The proof of Theorem~\ref{thm:fully_revealing} reveals that a fully revealing encoding policy is optimal among randomized policies when $u^s(i,i)\geq u^s(i,j)$ holds for all $i,j\in\{0,1,\dots,M-1\}$. A similar condition exists in Lemma~\ref{lem:fully_revealing_lemma} with an additional requirement on prior probabilities. The difference arises from the fact that the sender may convince the receiver to take the action $y=j$ under a source observation of $x=i$ with a randomized encoding policy regardless of the prior probabilities.
\end{remark}

Then, we consider noninformative encoding policies. 

\begin{lemma}\label{lem:noninf_lemma}
The optimal sender under a deterministic policy restriction is noninformative if and only if the following inequalities are satisfied:
\begin{align}
u^s(i,0)\geq u^s(i,i),
\label{eq:noninf_lemma_eq1}
\end{align}
for all $i\in\{1,2,\dots,M-1\}$, and
\begin{align}
\pi_iu^s(i,0)+\pi_ju^s(j,0)\geq \pi_iu^s(i,i)+\pi_ju^s(j,i)
\label{eq:noninf_lemma_eq2}
\end{align}
for all $i,j\in\{1,2,\dots,M-1\}$ satisfying $i<j$.
\end{lemma}

\begin{proof}
See Appendix~\ref{proof:noninf_lemma}.
\end{proof}

\begin{remark}
Under a noninformative policy, the receiver takes the action $y=0$ for all source observations at the sender. The optimality of such a policy among deterministic policies requires that for a source observation of $x=i$, the sender does not wish to replace the corresponding action with $y=i$ for all $i\neq 0$, which is feasible by revealing the corresponding source. In addition, the sender does not wish to construct a separate bin for sources $x=i$ and $x=j$ with $\pi_i\geq \pi_j$ to induce $y=i$ instead of $y=0$ for these sources. 
\end{remark}

\begin{theorem}\label{thm:noninf}
Suppose that the optimal sender under a deterministic policy restriction is noninformative with $\pi_0>\pi_i$ for all $i\in\{1,2,\dots,M-1\}$. Then, the equilibrium solution involves a randomized encoder if and only if there exists a critical posterior $Q$ that satisfies 
\begin{align}
\sum_{x\in\mathcal{X}}u^s(x,0)Q(x) < u^{s,*}(Q), 
\label{eq:noninf_thm_cond}
\end{align}
where $Q$ induces equal receiver preferences for at least two distinct actions. 
\end{theorem}

\begin{proof}
See Appendix~\ref{proof:noninf_thm}.
\end{proof}

Theorem~\ref{thm:noninf} demonstrates that even though the sender completely hides the source under a deterministic policy restriction, certain conditions convince the sender to reveal information related to the source when randomized policies are allowed. In other words, under these conditions, conveying distorted information is better than giving no information for the sender.

\begin{remark}
With a randomized policy, the sender is able to mix any source combination to confuse the receiver. Theorem~\ref{thm:noninf} shows that for the optimality of a noninformative policy among randomized policies, it is required that the sender does not gain by using such source mixtures to induce an action other than $y=0$ for any source combination. 
\end{remark}

\begin{lemma}\label{lem:quantized}
Suppose that $\pi_0>\pi_1>\dots>\pi_{M-1}$ holds. Then, the optimal sender under a deterministic policy restriction is a quantization policy with quantization bins $B_0,B_1,\dots,B_{n-1}$ and respective decisions $x_0,x_1,\dots,x_{n-1}$ for these bins if and only if the following hold:
\begin{enumerate}
\item [(i)] For all $i\in\{0,1,\dots,n-1\}$, we have that 
\begin{align}
u^s(x,x_i)\geq u^s(x,x)
\label{eq:quantized_lemma_eq1}
\end{align}
holds for all $x\in B_i$ with $x\neq x_i$. 
\item [(ii)] For all $i\in\{0,1,\dots,n-1\}$, we have that 
\begin{align}
\pi_xu^s(x,x_i) + \pi_y u^s(y,x_i) 
\geq 
\pi_x u^s(x,x)+ \pi_y u^s(y,x)
\label{eq:quantized_lemma_eq2}
\end{align}
is satisfied for all $x,y\in B_i$ with $x\neq x_i$, $y\neq x_i$ and $x<y$.
\item [(iii)] For all $i,j\in\{0,1,\dots,n-1\}$ with $i\neq j$, we have that 
\begin{align}
u^s(x,x_i) \geq u^s(x,x_j)
\label{eq:quantized_lemma_eq3}
\end{align}
holds for all $x\in B_i$ with $x\neq x_i$ and $x_j<x$.
\item [(iv)] For all $i,j\in\{0,1,\dots,n-1\}$ with $i\neq j$, the inequality
\begin{align}
\pi_x u^s(x,x_i) + \pi_yu^s(y,x_j) 
\geq 
\pi_x u^s(x,x) + \pi_yu^s(y,x) 
\label{eq:quantized_lemma_eq4}
\end{align}
is satisfied for all $x\in B_i$ and $y\in B_j$ with $x\neq x_i$, $y\neq x_j$ and $x<y$.
\item [(v)] For all $i,j\in\{0,1,\dots,n-1\}$ with $i\neq j$ and $x_i<x_j$, the optimal deterministic policy considering the sources $\{x_i\}\cup B_j$ is a quantization policy with bins $\{x_i\}$ and $B_j$. 
\end{enumerate}
\end{lemma}

\begin{proof}
See Appendix~\ref{proof:quantized_lemma}.
\end{proof}

\begin{remark}
While the conditions in (i) and (ii) are due to Lemma~\ref{lem:noninf_lemma} applied to each quantization bin, the remaining conditions can be interpreted as follows. The conditions in (iii) ensure that the sender does not wish to change the bin of a source whose prior probability is not the largest in the bin. Under the conditions in (iv), the sender does not wish to construct a separate bin by using sources from multiple bins whose prior probabilities are not the largest in the corresponding bins. Finally, the conditions in (v) ensure that the sender does not wish to change the action corresponding to the source with the largest probability in a bin. In this case, it is required to take the other sources into account since the corresponding actions are affected by the bin change of the source with the largest prior probability.  
\end{remark}

\begin{theorem}\label{thm:quantized}
Suppose that the optimal encoding policy among deterministic policies is a quantization policy with bins $B_0,B_1,\dots,B_{n-1}$ and respective decisions $x_0,x_1,\dots,x_{n-1}$, where $n\geq 2$. Suppose that $\pi_0>\pi_1>\dots>\pi_{M-1}$ holds. Then, the equilibrium solution involves a randomized encoder if and only if at least one of the following conditions holds:
\begin{enumerate}
\item [(i)] There exists a critical posterior $Q$ that satisfies 
\begin{align}
\sum_{x\in\mathcal{X}}u^s(x,x_i)Q(x) < u^{s,*}(Q), 
\label{eq:quantized_thm_eq1}
\end{align}
where the receiver is indifferent only between receiver actions in $\mathcal{A}$ under $Q$ for some $\mathcal{A}\subseteq B_i$ with $i\in\{0,1,\dots,n-1\}$.

\item [(ii)] There exists a critical posterior $Q$ that induces equal receiver preferences for actions from multiple bins with 
\begin{align}
\sum_{i=0}^{M-1} \alpha_i u^{s,*}(Q_i)
<u^{s,*}(Q),
\label{eq:quantized_thm_eq2}
\end{align}
where $Q_i(x_i)=1$ for $i=0,1,\dots,n-1$, the set $\{Q_{n},Q_{n+1},\dots,Q_{M-1}\}$ consists of critical posteriors inducing equal preferences only for two distinct actions including $x_i$ inside the bin $B_i$ for $i\in\{0,1,\dots,n-1\}$, and $\alpha_i$ is the $i$th entry of $\boldsymbol{\alpha}$ with
\begin{align}
\boldsymbol{\alpha} = 
\begin{bmatrix}
Q_0 & Q_1 & \cdots &Q_{M-1}
\end{bmatrix}^{-1} Q.
\label{eq:quantized_thm_eq3}
\end{align}
\end{enumerate}
\end{theorem}

\begin{proof}
See Appendix~\ref{proof:quantized_thm}.
\end{proof}

\begin{remark}
In Lemma~\ref{lem:quantized}, the condition in (i) is due to Theorem~\ref{thm:noninf} applied to each bin. In other words, if at least one of the conditions in (i) is satisfied, then the sender's payoff improves by applying randomization for sources within a single bin. On the other hand, if the conditions in (i) are not satisfied, one needs to consider also randomization of sources from multiple bins to see whether randomization is beneficial for the sender. The conditions in (ii) addresses such randomizations at the sender. In particular, such randomized policies that mix sources from multiple bins can be obtained from the optimal quantization policy by adjusting probabilities based on $\boldsymbol{\alpha}$ specified in \eqref{eq:quantized_thm_eq3}. Then, the inequality in \eqref{eq:quantized_thm_eq2} ensures that the sender's payoff improves by mixing the corresponding sources. If the inequalities in \eqref{eq:quantized_thm_eq2} are not satisfied for each source mixture from multiple bins, then the equilibrium solution is deterministic with the corresponding quantization policy. 
\end{remark}

\section{Explicit Solutions and Analysis for Binary and Ternary Settings}

\subsection{Binary Setting}\label{sec:binary}

In this section, we consider a binary source and a binary receiver action by taking $\mathcal{X}=\mathcal{Y}=\{0,1\}$. We let $\pi_0\geq \pi_1$. Without any loss, we can take $\mathcal{M}=\{0,1\}$. In the following, we investigate the equilibrium behavior by applying the results in Lemma~\ref{lem:fully_revealing_lemma}, Lemma~\ref{lem:noninf_lemma}, Theorem~\ref{thm:fully_revealing} and Theorem~\ref{thm:noninf}. Since the only deterministic policies in the binary case are fully revealing and noninformative policies, the results of Lemma~\ref{lem:quantized} and Theorem~\ref{thm:quantized} are not applicable in this case.

\begin{corollary}\label{cor:binary}
If $\pi_0>\pi_1$ holds, the equilibrium solution can be obtained by the following:  
\begin{enumerate}
\item[(i)] If $\min\{u^s(1,1)-u^s(1,0),u^s(0,0)-u^s(0,1)\}\geq 0 $ holds, the equilibrium solution involves a fully revealing sender.
\item[(ii)] If $u^s(1,0)-u^s(1,1)>0$ and $u^s(1,0)-u^s(1,1)\geq u^s(0,1)-u^s(0,0)$ are satisfied, the equilibrium solution involves a noninformative sender.
\item[(iii)] If $\min\{u^s(1,1)-u^s(1,0),u^s(0,1)-u^s(0,0)\}> 0$ or $u^s(0,1)-u^s(0,0)>u^s(1,0)-u^s(1,1)>0$ holds, the equilibrium solution is attained by a randomized sender.
\end{enumerate}
For $\pi_0=\pi_1$, if $\min\{u^s(1,1)-u^s(1,0),u^s(0,0)-u^s(0,1)\}>0$ holds, the equilibrium solution leads to a fully revealing sender, whereas in the converse case, the equilibrium solution involves a noninformative sender.
\end{corollary}

\begin{remark}
Corollary~\ref{cor:binary} reveals that for $\pi_0>\pi_1$, the sender fully reveals the source when it does not prefer a mismatched action under both source realizations. A noninformative sender is optimal when the sender prefers a mismatched action at the receiver under a source observation of $x=1$ and the desire to induce a mismatched action for the source $x=1$ outweighs such a desire for the source $x=0$. In the remaining cases, we obtain a randomized sender as the equilibrium solution where the sender gives uncertain information to the receiver by inducing a posterior of $[1/2,1/2]^T$. Moreover, it is noted that a randomized sender is not needed for $\pi_0=\pi_1$ since the receiver is already uncertain regarding the source under the prior distribution. 
\end{remark}

\begin{remark}[Interpretation of the conditions and results.] \label{InterpRemark} Theorem~\ref{thm:algorithm} implies that the equilibrium solution induces either of the following posterior pairs depending on the payoff structure and the prior probability of the source: (i) $Q_0=[1,0]^T$ and $Q_1=[0,1]^T$, (ii) $Q_0=Q_1=[\pi_0,\pi_1]^T$, (iii) $Q_0=[1,0]^T$ and $Q_1=[1/2,1/2]^T$, (iv) $Q_0=[1/2,1/2]^T$ and $Q_1=[0,1]^T$. We discuss these cases in the following.

\begin{enumerate}
\item [(i)] In this case, the sender does not wish to hide information from the receiver. In other words, even though the payoffs are not aligned, there is full transparency at the sender. Since posteriors $Q_0=[0,1]^T$ and $Q_1=[1,0]^T$ are induced, the receiver is certain of the source for each message observation. The equilibrium is attained by such a policy in Fig.~\ref{fig:binary_fig4}.
\item [(ii)] In this case, the equilibrium solution essentially leads to a single posterior distribution that is equal to the prior distribution. This corresponds to a noninformative policy, which is the essentially unique (lossy/irreversible) quantization policy for the binary source setting; which is also equivalent to a noninformative randomized encoding policy (Definition~\ref{def:deterministic}(iii)). In this case, the sender's payoff does not improve by inducing any posterior other than the prior and the sender completely hides information from the receiver. 
\item [(iii)] In this case, the sender applies randomization. When the equilibrium solution is randomized, inducing a posterior other than $[0,1]^T$, $[1,0]^T$ and $[\pi_0,\pi_1]^T$ improves the sender's payoff. Due to Theorem~\ref{thm:algorithm}, a randomized equilibrium solution always leads to a posterior of $[1/2,1/2]^T$ under which the receiver is indifferent between the actions. By anticipating and inducing the posterior under which the receiver is indifferent, the sender manipulates the receiver to take the desired action. In order to induce these posteriors, the sender maps an observation of $X=0$ to messages $M=0$ and $M=1$ with certain nonzero probabilities, and transmits a message of $M=1$ when $X=1$ is observed. When the receiver observes $M=0$, the receiver knows that the source is $X=0$. In that sense, the sender is telling the truth in this case. However, a message observation of $M=1$ means that the source is either $X=0$ or $X=1$ each with a probability of $1/2$. While the receiver is trying to estimate the source, the sender transmits a message that prevents the receiver from extracting the source. In fact, the sender gives corrupted information to the receiver in order to take advantage of such uncertainty at the receiver. In Fig.~\ref{fig:binary_fig1}, the equilibrium solution is attained by such a randomized policy. 
\item [(iv)] This case is similar to the previous case. In particular, the sender maps an observation of $X=1$ to messages $M=0$ and $M=1$ with certain nonzero probabilities, and conveys a message of $M=0$ when $X=0$ is observed. Again, information is corrupted that misleads the receiver under the message observation of $M=0$. 
\end{enumerate}
\end{remark}

\begin{figure}
\centering
\includegraphics[width=2.5in]{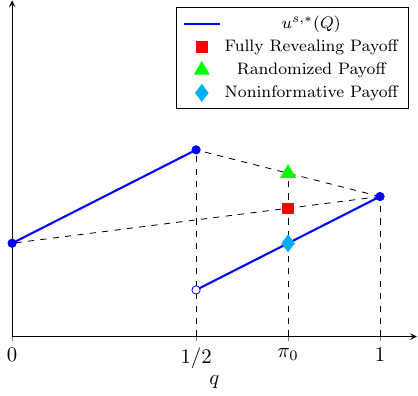}
\caption{Payoffs of fully informative, randomized and noninformative policies for a particular payoff structure where $Q=[q,(1-q)]^T$. A fully informative policy induces posteriors of $[0,1]^T$ and $[1,0]^T$ with probabilities of $\pi_1$ and $\pi_0$, respectively. A noninformative policy is equivalent to inducing a single posterior $[\pi_0,(1-\pi_0)]^T$. The particular randomized policy leads to the equilibrium solution by inducing posteriors of $[0,1]^T$ and $[1/2,1/2]^T$ with certain probabilities. In order to obtain these posteriors, the encoder maps a source observation of $X=1$ to two distinct messages (i.e., $M=0$ and $M=1$) with nonzero probabilities while mapping an observation of $X=0$ to a single message, say $M=0$, with probability 1. A different payoff structure with the values at $q=1/2$ fixed may lead to a fully revealing policy as the equilibrium solution whereas a noninformative policy cannot be the equilibrium solution in this case.}
\label{fig:binary_fig1}
\end{figure}

\begin{figure}
\centering
\includegraphics[width=2.5in]{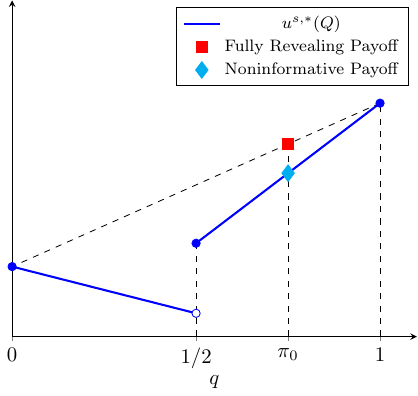}
\caption{Payoffs of fully informative and noninformative policies for a particular payoff structure where $Q=[q,(1-q)]^T$. For this particular payoff structure, the equilibrium solution is attained by fully revealing the source. In this setup, an encoding policy that induces posteriors of $[0,1]^T$ and $[1/2,1/2]^T$ is equivalent to a noninformative policy. A different payoff structure with the values at $q=1/2$ fixed may lead to a noninformative policy as the equilibrium solution whereas there is no need to use a randomized policy in this case, as well.}
\label{fig:binary_fig4}
\end{figure}

\subsection{Ternary Setting}

In this section, we consider a ternary setting and discuss implications of the theoretical results in the previous section. Let $\mathcal{X}=\mathcal{Y}=\{0,1,2\}$. Without any loss, we let $\mathcal{M}=\{0,1,2\}$. In addition, we consider $\pi_0\geq \pi_1\geq \pi_2$. Due to Theorem~\ref{thm:algorithm}, only the posteriors that are illustrated in Fig.\ref{fig:ternary} are necessary for obtaining the equilibrium solution. In other words, the sender does not gain by inducing a posterior other than the ones in Fig.\ref{fig:ternary}.

\begin{figure}
\centering
\includegraphics[width=2.5in]{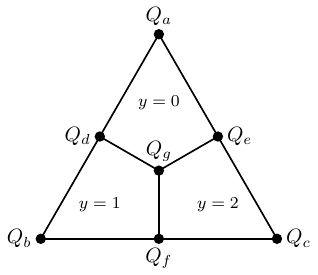}
\caption{Illustration of posteriors that can be induced by an optimal encoding policy for the ternary setting where $Q_a=[1,0,0]^T$, $Q_b=[0,1,0]^T$, $Q_c=[0,0,1]^T$, $Q_d=[1/2,1/2,0]^T$, $Q_e=[1/2,0,1/2]^T$, $Q_f=[0,1/2,1/2]^T$ and $Q_g=[1/3,1/3,1/3]^T$.}
\label{fig:ternary}
\end{figure}

\begin{remark}
Under the conditions of Lemma~\ref{lem:fully_revealing_lemma}, the optimal encoder under a deterministic policy restriction induces posteriors $Q_a$, $Q_b$ and $Q_c$ with probabilities $\pi_0$, $\pi_1$ and $\pi_2$, respectively. Due to Theorem~\ref{thm:fully_revealing}, one needs to compare the payoffs under posterior pairs $(Q_a,Q_b)$, $(Q_a,Q_c)$ and $(Q_b,Q_c)$ with $Q_d$, $Q_e$ and $Q_f$, respectively, to see whether randomization improves the sender's payoff. This comparison effectively checks whether giving misinformation by mixing source pairs is beneficial for the sender. The proof of Theorem~\ref{thm:fully_revealing} reveals that the sender does not wish to induce $Q_g$ by mixing all the sources as long as pairwise randomization is not beneficial for all pairs.
\end{remark}

\begin{remark}
If the conditions in Lemma~\ref{lem:noninf_lemma} are satisfied, the equilibrium solution under a deterministic policy restriction is equivalent to inducing a single posterior of $[\pi_0,\pi_1,\pi_2]^T$. A randomized sender outperforms such a noninformative policy if inducing at least one of the posteriors in $\{Q_d,Q_e,Q_f,Q_g\}$ yields a larger payoff for the sender. This means that it is required to take all source mixtures into account to conclude that the optimal encoder is deterministic or randomized.  
\end{remark}

\begin{remark}
As an example, suppose that the equilibrium solution under a deterministic policy restriction quantizes the source pair of $x=0$ and $x=1$ while revealing the source value of $x=2$ (see Lemma~\ref{lem:quantized} for the corresponding conditions). Such a quantized policy leads to posteriors of $[\pi_0/(\pi_0+\pi_1),\pi_1/(\pi_0+\pi_1),0]^T$ and $[0,0,1]^T$. Theorem~\ref{thm:quantized} specifies explicit conditions ensuring the optimality of a randomized or quantized encoder. In particular, one first needs to compare the actions of $y=0$ and $y=1$ under $Q_d$ (see the condition in (i)). If the corresponding condition in (i) is not satisfied, then the quantization policy is equivalent to inducing posteriors of $Q_a$, $Q_c$ and $Q_d$ with certain probabilities. Then, the conditions in (ii) check whether inducing at least one of the posteriors in $\{Q_e,Q_f,Q_g\}$ improves the sender's payoff, where one changes the probabilities of the posteriors in $\{Q_a,Q_c,Q_d\}$ specified by the coefficients in \eqref{eq:quantized_thm_eq3}. 
\end{remark}

\section{Conclusion}

We have investigated the structure of the equilibrium solution for the Bayesian persuasion problem. We have completely characterized conditions under which the equilibrium solution leads to randomized, fully revealing, quantized, or noninformative signaling at the sender, considering a setup with a receiver that wishes to extract the source. We have provided an algorithm to obtain the equilibrium solution where one restricts attention to a finite set of posteriors at the receiver without any loss of optimality. We have also proven that for general payoff structures, the equilibrium solution leads to noninformative signaling when the payoffs are completely misaligned. In addition, we have unified the completely aligned (team) and misaligned (zero-sum) setups for the specific receiver payoff structure by showing that the continuity of the sender's value function leads to either a fully revealing or a noninformative sender. Future work includes considering general receiver payoff structures.

\section*{Acknowledgments}
The authors are grateful to Dr. Serkan Sar{\i}ta\c{s} for his constructive comments. 

\appendix

\subsection{Proof of Lemma~\ref{lem:nonconvex_posteriors}}\label{proof:nonconvex_posteriors}

We provide an example to conclude the result. Consider a ternary source setting. Let $\pi=[0.5,0.4,0.1]^T$. A set of posteriors $Q_0=[1,0,0]^T$, $Q_1=[0,1,0]^T$ and $Q_2=[0,0,1]^T$ is feasible with $P(Q_0)=0.5$, $P(Q_1)=0.4$ and $P(Q_2)=0.1$. Let us take another set of posteriors $Q_0'=[0.5,0.45,0.05]^T$, $Q_1'=[0.4,0.1,0.5]^T$ and $Q_2'=[0.55,0.05,0.4]^T$ that also satisfy the Bayesian plausibility constraint in \eqref{eq:optim_problem} with certain probabilities. Then, it is a linear algebra exercise to show that $\{Q_0'',Q_1'',Q_2''\}$ is not a feasible posterior distribution combination where $Q_0''=\lambda Q_0+ (1-\lambda)Q_i'$, $Q_1''=\lambda Q_1+ (1-\lambda)Q_j'$, $Q_2''=\lambda Q_2+ (1-\lambda)Q_k'$, $\lambda = 1/2$, and $(i,j,k)$ is any permutation of $(0,1,2)$. 

\subsection{Proof of Theorem~\ref{thm:blackwell}}\label{proof:blackwell}

Let $u^s(x,y)=u^r(x,y)\triangleq u(x,y)$. For a message observation $m\in\mathcal{M}$, the receiver computes the posterior distribution of the source to determine its action. Let $Q(x)$ denote a posterior distribution of the source where $x\in\mathcal{X}$. For a given posterior distribution $Q(\cdot)$, we define the optimal payoff value as
\begin{align}
V(Q) = 
\max_{y\in\mathcal{Y}} 
\sum_{x\in\mathcal{X}} 
u(x,y) Q(x).\label{eq:VQ}
\end{align}
We first show that $V(Q)$ is convex. Towards that goal, let us take two probability distributions $Q'$ and $Q''$. Let $Q=\alpha Q' + (1-\alpha) Q''$ for $\alpha\in[0,1]$. Then, we write
\begin{align}
V(Q) 
&= \max_{y\in\mathcal{Y}} 
\Big(
\alpha \sum_{x\in\mathcal{X}} 
u(x,y) Q'(x)\nonumber \\
&\hphantom{= \max_{y\in\mathcal{Y}} \Big(}
+(1-\alpha) \sum_{x\in\mathcal{X}} 
u(x,y) Q''(x)
\Big)\nonumber \\
&\leq  \max_{y\in\mathcal{Y}}
\alpha \sum_{x\in\mathcal{X}} 
u(x,y) Q'(x)\nonumber \\
&\hphantom{=}
+ \max_{y\in\mathcal{Y}}
(1-\alpha) \sum_{x\in\mathcal{X}} 
u(x,y) Q''(x)\nonumber \\
&= \alpha V(Q') + (1-\alpha) V(Q''),
\end{align}
which proves that $V(\cdot)$ is convex. 

Next, we compare two encoding policies as follows. Let us consider an encoding policy $\gamma^s_1(\cdot)$ defined by conditional distributions $P(m^1 | x)$ where $m^1\in\mathcal{M}$ and $x\in\mathcal{X}$. These conditional distributions form an information structure from the sender to the receiver. Let us consider another information structure induced by a certain encoding policy $\gamma^s_2(\cdot)$ with the following property: 
\begin{align}
P(m^2\vert x) = \sum_{m^1\in\mathcal{M}} P(m^2\vert m^1) P(m^1\vert x)
\label{eq:garbling}
\end{align}
for some conditional distributions $P(m^2|m^1)$ where $m^1\in\mathcal{M}$, $m^2\in\mathcal{M}$ and $x\in\mathcal{X}$. Under \eqref{eq:garbling}, we say that the information structure leading to $P(m^2|x)$ is garbled (or stochastically degraded) with respect to the information structure leading to $P(m^1|x)$. Next, we compare the payoffs under these information structures in the following:
\begin{align}
U(\gamma^s_2) 
&= \sum_{m^2\in\mathcal{M}} 
V\Big(P(\cdot\,\vert m^2)\Big) 
P(m^2) \nonumber \\
&= \sum_{m^2\in\mathcal{M}} 
V\bigg(\sum_{m^1\in\mathcal{M}}P(\cdot\,\vert m^1) P(m^1\vert m^2)\bigg) 
P(m^2) \nonumber \\
&\leq  \sum_{m^2\in\mathcal{M}} 
\sum_{m^1\in\mathcal{M}} P(m^1\vert m^2)
V\Big(P(\cdot\,\vert m^1)\Big) 
P(m^2) \nonumber \\
&=  \sum_{m^1\in\mathcal{M}} 
V\Big(P(\cdot\,\vert m^1)\Big) 
P(m^1) =U(\gamma^s_1),
\end{align}
where the second line follows from the fact that $X$, $M^1$ and $M^2$ form a Markov chain in that order, and the inequality is due to the convexity of $V(\cdot)$ and Jensen's inequality. This proves that if one applies stochastic degradation to an information structure, then the payoff cannot improve. Therefore, it is optimal to completely reveal the source since any other encoding policy induces an information structure that is a garbling of a fully revealing information structure. 

\subsection{Proof of Theorem~\ref{thm:zerosum}}\label{proof:zerosum}

We show that $u^{s,*}(\cdot)$ is concave to prove the result. To achieve this goal, we make use of the following observation. Suppose that $y\in\mathcal{Y}$ is a best response of the receiver under a posterior distribution $Q$. This means that $\sum_{x\in\mathcal{X}}u^r(x,y)Q(x)\geq \sum_{x\in\mathcal{X}}u^r(x,y')Q(x)$ holds for all $y'\in\mathcal{Y}$. Due to the zero-sum property, we obtain $\sum_{x\in\mathcal{X}}u^s(x,y)Q(x)\leq \sum_{x\in\mathcal{X}}u^s(x,y')Q(x)$ for all $y'\in\mathcal{Y}$. This implies that $y$ is the worst receiver action from the perspective of the sender under the posterior distribution $Q$. Now, let us take $Q'$ and $Q''$ that respectively lead to $y=y'$ and $y=y''$ as the best response of the receiver, where $y',y''\in\mathcal{Y}$. Let $Q(x) = \lambda Q'(x) + (1-\lambda)Q''(x)$ for $\lambda \in [0,1]$. If the best response of the receiver under $Q$ is $y\in\mathcal{Y}$, we can write 
\begin{align}
&\lambda \sum_{x\in\mathcal{X}} u^s(x,y') Q'(x) 
+(1-\lambda) \sum_{x\in\mathcal{X}} u^s(x,y'') Q''(x) \nonumber\\
&\leq 
\sum_{x\in\mathcal{X}} u^s(x,y) (\lambda Q'(x) +(1-\lambda)Q''(x)), 
\label{eq:zerosum}
\end{align}
where we use the fact that $y'$ and $y''$ are the least preferred actions for the sender under the posterior distributions $Q'$ and $Q''$, respectively. Hence, the inequality in \eqref{eq:zerosum} reveals that $u^{s,*}(\cdot)$ is concave. Then, by Theorem~\ref{thm:KamGen}, it follows that the equilibrium solution involves a noninformative sender. 

\subsection{Proof of Lemma~\ref{lem:continuity}}\label{proof:lem_continuity}

Let us take a set $A\subseteq \mathcal{Y}$ and a posterior $Q$ such that \eqref{eq:lem_continuity_eq1} is satisfied. Under this posterior, the receiver chooses the action that gives the highest payoff for the sender among the actions in $A$. Let $y'$ denote this action. If the condition in \eqref{eq:lem_continuity_eq2} is not satisfied, then there exists an action $y''\in A$ with 
\begin{align}
\sum_{x\in\mathcal{X}}u^s(x,y')Q(x)>\sum_{x\in\mathcal{X}}u^s(x,y'')Q(x).
\label{eq:lem_continuity_eq3}
\end{align}
Now, consider $Q_\epsilon$ with $Q_\epsilon(x)=Q(x)-\frac{\epsilon}{|A|-1}$ for all $x\neq y''$ and $x\in A$, $Q_\epsilon(y'')=Q(y'')+\epsilon$, and $Q_\epsilon(x)=0$ for all $x\notin A$, where $\epsilon>0$ is small. Since $Q(x)=1/|A|$ is satisfied for all $x\in A$, $Q_\epsilon$ is a valid distribution for small $\epsilon>0$. Under $Q_\epsilon$, the receiver selects $y''$ as its best response. This implies that $u^{s,*}(Q)>\lim_{\epsilon\to 0^+} u^{s,*}(Q_\epsilon)$ due to \eqref{eq:lem_continuity_eq3}. This reveals that if \eqref{eq:lem_continuity_eq2} is not satisfied, the continuity condition is violated. The proof for the converse is similar.

\subsection{Proof of Theorem~\ref{thm:continuous}}\label{proof:continuous}

For the given payoff structure of the receiver (i.e., $u^r(x,y)=1$ for all $x=y$ and $u^r(x,y)=0$ for all $x\neq y$), we obtain a set of posterior probability distributions under which the receiver is indifferent between certain actions. Under these posteriors, the sender must also be indifferent between the corresponding receiver actions since otherwise the continuity condition is violated (see Lemma~\ref{lem:continuity}). Although this set of posterior distributions is large, we show that one can restrict attention to a subset of them. More specifically, it is sufficient to consider $Q_a\triangleq [1/M,\dots,1/M]^T$ and $\{Q_i\}_{i=0}^{M-1}$ where $Q_i(x)=1/(M-1)$ for $x\neq i$ and $Q_i(i)=0$. Here, $Q_a$ is a posterior under which the receiver is indifferent between all the actions, and $Q_i$ is a posterior under which the receiver is indifferent between all the actions except for the $i$th action. The sender must also be indifferent between all the receiver actions under $Q_a$. More specifically, this requires that
\begin{align}
\sum_{x\in\mathcal{X}} Q_a(x)u^s(x,i)=\sum_{x\in\mathcal{X}} Q_a(x)u^s(x,j)
\label{eq:continuity_1}
\end{align} 
holds for all $i,j\in\mathcal{Y}$. This gives $M-1$ linearly independent equations that the sender's payoff structure must satisfy. In addition, the set $\{Q_i\}_{i=0}^{M-1}$ contains posteriors under which the receiver prefers all the actions equally except for the $i$th action. Since the sender must also be indifferent between the corresponding actions for continuity of $u^{s,*}(\cdot)$, we obtain
\begin{align}
\sum_{x\in\mathcal{X}} Q_i(x)u^s(x,j) = \sum_{x\in\mathcal{X}} Q_i(x)u^s(x,k) 
\label{eq:continuity_2}
\end{align}
for all $j,k\in\mathcal{Y}$ satisfying $j\neq i$ and $k\neq i$. This gives a total of $M(M-2)$ linearly independent equations that the sender's payoff structure must satisfy. 

Next, we show that the conditions in \eqref{eq:continuity_1} and \eqref{eq:continuity_2} are sufficient for ensuring the continuity of $u^{s,*}(\cdot)$. First of all, we have a family of posteriors $\lambda Q_i + (1-\lambda)Q_a$ parametrized by $\lambda\in(0,1)$ under which the receiver is indifferent between all the receiver actions except for the $i$th action. By combining \eqref{eq:continuity_1} and \eqref{eq:continuity_2}, it is seen that the sender is also indifferent between the corresponding actions under $\lambda Q_i + (1-\lambda)Q_a$. Hence, as long as \eqref{eq:continuity_1} and \eqref{eq:continuity_2} are satisfied, $u^{s,*}(\cdot)$ is continuous at posteriors with all nonzero entries. Moreover, one can iteratively obtain the continuity conditions for posteriors with at least one nonzero entry by using \eqref{eq:continuity_1} and \eqref{eq:continuity_2}. For brevity, we give an example that illustrates the approach for obtaining the rest of the continuity conditions. More specifically, the receiver is indifferent between all the receiver actions except for $y=M-2$ and $y=M-1$ under $Q=[1/(M-2),\dots,1/(M-2),0,0]^T$. The sender must also be indifferent between the corresponding actions under $Q$. These indifference conditions involving the sender's payoff under $Q$ can be obtained by linearly combining the conditions under $Q_a$, $Q_{M-2}$ and $Q_{M-1}$ in \eqref{eq:continuity_1} and \eqref{eq:continuity_2}. In particular, we know that the conditions in \eqref{eq:continuity_1} for $Q_a$ and the conditions in \eqref{eq:continuity_2} for $Q_{M-2}$ and $Q_{M-1}$ are satisfied for all actions in $\{0,1,\dots,M-2\}$. In addition, we can write 
\begin{align}
Q = \big((M-1)(Q_{M-1}+Q_{M-2}) - M Q_a\big)/(M-2).
\label{eq:continuity_3}
\end{align}
Then, we can use the corresponding coefficients in \eqref{eq:continuity_3} to linearly combine the equations involving $Q_a$, $Q_{M-2}$ and $Q_{M-1}$ in \eqref{eq:continuity_1} and \eqref{eq:continuity_2} so that we obtain the continuity conditions involving $Q$ for all actions in $\{0,1,\dots,M-2\}$. In a similar manner, one can linearly combine the equations in \eqref{eq:continuity_1} and \eqref{eq:continuity_2} to express the continuity conditions for all posteriors inducing equal receiver preferences for certain actions. This reveals that it is sufficient to impose the constraints in \eqref{eq:continuity_1} and \eqref{eq:continuity_2} in order to make $u^{s,*}(\cdot)$ continuous.

We have shown that the equalities in \eqref{eq:continuity_1} and \eqref{eq:continuity_2} are necessary and sufficient conditions for the continuity of $u^{s,*}(\cdot)$. Thus, the sender's payoff structure must satisfy $M^2-M-1$ linear equations that are linearly independent. These equations have at least a solution since the receiver's payoff structure already satisfies them. With $M^2$ unknowns for the sender's payoff structure and $M^2-M-1$ linearly independent equations, we obtain a family of solutions with dimension $M+1$. In addition, we know that for each $x\in\mathcal{X}$, a payoff structure satisfying $u^s(x,0)=u^s(x,1)=\dots=u^s(x,M-1)\neq 0$ and $u^s(x',y)=0$ for all $x'\neq x$ and $y\in\mathcal{Y}$ is a solution. These payoff structures give $M$ distinct and linearly independent elements to construct a basis for the space of solutions. Since the space of solutions has dimension $M+1$, we have an additional element that makes up a complete basis for the family of solutions, ensuring the continuity of $u^{s,*}(\cdot)$. These observations reveal that the sender's payoff structure must be of the following form:  
\begin{align*}
\begin{bmatrix}
u^s(0,0) & u^s(1,0) & \dots \\
u^s(0,1) & u^s(1,1) & \dots \\
\vdots & \vdots &  \\
\end{bmatrix}
=
\begin{bmatrix}
a_0 & a_1 & \dots \\
a_0 & a_1 & \dots \\
\vdots & \vdots & 
\end{bmatrix}
+
\alpha
W
\end{align*}
for a certain specific matrix $W$ and arbitrary $a_0,a_1,\dots\in\mathbb{R}$ and $\alpha\in\mathbb{R}$, where the matrix comes from the additional element contributing to the basis. Similarly, 
the receiver's payoff structure must satisfy 
\begin{align*}
\begin{bmatrix}
u^r(0,0) & u^r(1,0) & \dots \\
u^r(0,1) & u^r(1,1) & \dots \\
\vdots & \vdots &  \\
\end{bmatrix}
=
\begin{bmatrix}
b_0 & b_1 & \dots \\
b_0 & b_1 & \dots \\
\vdots & \vdots & 
\end{bmatrix}
+\beta W
\end{align*}
for some $b_0,b_1,\dots\in\mathbb{R}$ and $\beta\in\mathbb{R}$, where the matrix $W$ is as in the sender's payoff. Since we consider a setting where the receiver's payoff is given by $u^r(x,y)=1$ for $x=y$ and $u^r(x,y)=0$ for $x\neq y$, we obtain $b_0=b_1=\dots=0$, $\beta=1$ and $W=\mathrm{diag}([1,\dots,1]^T)$. This reveals that the problem is essentially equivalent to either a team theoretic setup or a zero-sum setup, depending on the sign of $\alpha$. Then, the result follows from Theorem~\ref{thm:blackwell} and Theorem~\ref{thm:zerosum}.

\subsection{Proof of Theorem~\ref{thm:algorithm}}\label{proof:algorithm}

While the set of feasible posterior combinations satisfying the Bayesian plausibility condition is large, we show that it is sufficient to consider the posteriors in the set $\mathcal{Q}$. To that end, we first show the convexity of the decision regions from the perspective of the receiver. In particular, we prove that the set of posteriors under which a particular action $y$ gives the largest (or equal to) payoff for the receiver is convex. Towards that goal, let us take two posteriors $Q'$ and $Q''$ that satisfy this property. Let $Q=\lambda Q' +(1-\lambda)Q''$ for $\lambda\in(0,1)$. Then, we write
\begin{align*}
&\sum_{x\in\mathcal{X}} u^{r}(x,y)(\lambda Q'(x) +(1-\lambda)Q''(x))\\
&\geq 
\sum_{x\in\mathcal{X}} u^{r}(x,y') (\lambda Q'(x) +(1-\lambda)Q''(x))
\end{align*}
for all $y'\neq y$, where the inequality is established by decomposing the summations and then using the fact that the receiver prefers the action $y$ under $Q'$ and $Q''$. This reveals that the action $y$ yields the largest payoff for the receiver under the convex combination $Q$, as well. 

Then, let us take a posterior $Q$ under which the action $y$ yields the largest payoff. By Krein–Milman Theorem \cite{Rudin1991}, we can express this posterior as a convex combination of the extreme points of the region where the action $y$ leads to the largest payoff for the receiver. The first step of the algorithm constructs a set $\mathcal{Q}$ that consists of such extreme points. In the following, we explain this construction. To begin with, the posteriors in the set $\{e_0,e_1,\dots,e_{M-1}\}$ are indeed extreme points since it is not possible to write them as a convex combination of any other probability distributions. Next, we consider posteriors under which the action $y$ leads to a strictly larger payoff than any other action. More specifically, let us take a posterior $Q$ under which $\sum_{x\in\mathcal{X}} u^r(x,y)Q(x)>\sum_{x\in\mathcal{X}} u^r(x,y')Q(x)$ holds for all $y'\neq y$. In this case, this posterior is either in the interior of the region where the receiver prefers $y$, or is a point at a boundary of the probability simplex. If $Q$ is in the interior, then this posterior is obviously not an extreme point. On the other hand, if $Q$ is at a boundary of the probability simplex, then one can verify that $Q$ can be represented as a convex combination of $e_y$ and the posteriors inducing equal preferences for certain actions including $y$. Therefore, it is concluded that we need to consider equal preferences to obtain the remaining extreme points. More specifically, for each $\mathcal{A}\subseteq \mathcal{X}$ with at least two elements, we obtain a posterior under which the receiver is indifferent between the actions in $\mathcal{A}$. In particular, due to the structure of the receiver's payoff, for a given $\mathcal{A}\subseteq \mathcal{X}$, the receiver prefers the actions in the set $\mathcal{A}$ equally and prefers these actions over any other action under the posterior distribution $Q$ with $Q(x)=1/|\mathcal{A}|$ for all $x\in\mathcal{A}$, and $Q(x)=0$ for all $x\notin\mathcal{A}$.

Next, let us take a posterior $Q$ under which the action $y$ yields the largest payoff for the receiver. The previous analysis reveals that this posterior can be written as a convex combination of certain posteriors in the set $\mathcal{Q}$. Now, we show that using this convex combination instead of $Q$ yields a larger or equal payoff for the sender. First, suppose that the receiver strictly prefers $y$ over any other action under $Q$. Due to the previous discussion, we can write $Q=\sum_{i=0}^{K-1}\lambda_iQ_i$ for some $K$ where $\lambda_i$'s with $0\leq \lambda_i\leq 1$ denote convex combination coefficients and $\{Q_i\}_{i=0}^{K-1}$ consists of $e_y$ and the posteriors in $\mathcal{Q}$ inducing equal receiver preferences for certain actions including $y$. For a posterior under which the receiver is indifferent between certain actions, the receiver picks the action that gives the largest payoff for the sender. This implies that if at least one of the posteriors in the convex combination leads to an action other than $y$, then the sender's payoff improves. More specifically, we obtain
\begin{align*}
u^{s,*}(Q)=\sum_{x\in\mathcal{X}}u^s(x,y)Q(x) \leq \sum_{i=0}^{K-1} \lambda_i u^{s,*}(Q_i),
\end{align*}
where the inequality follows from the fact the action $y$ is among the receiver's preferred actions under each $Q_i$. Next, suppose that the receiver is indifferent between certain actions and $y$ under $Q$. Let $\mathcal{C}$ denote this set of actions. We know that $Q$ is either in the set $\mathcal{Q}$ or can be represented as a convex combination of certain posteriors in the set $\mathcal{Q}$ under which the receiver prefers the actions in $\mathcal{C}$ equally. In the latter case, under a posterior in the convex combination, the receiver is indifferent between all the actions in $\mathcal{D}$, where $\mathcal{D}$ is some action set satisfying $\mathcal{C}\subseteq \mathcal{D}$. For such a posterior, the receiver acts according to the sender's preference among the actions in $\mathcal{D}$. Hence, if the sender prefers the same action for all posteriors in the convex combination, then using the convex combination yields the same payoff. On the other hand, if one posterior in the convex combination leads to a different action, then a higher payoff is obtained by using the convex combination. Therefore, it is seen that there is no need to induce a posterior that is not in the set $\mathcal{Q}$.

In the second step of the algorithm, we take $M$ linearly independent posteriors from the set $\mathcal{Q}$, where these posteriors satisfy the Bayesian plausibility constraint in \eqref{eq:algo_eq3}. In the following, we show that it is sufficient to consider only linearly independent posteriors. Towards that goal, let us take a feasible posterior combination $Q_0,Q_1,\dots,Q_{I-1}\in\mathcal{Q}$. Suppose that this posterior combination is linearly dependent. Let us take a set of positive coefficients $\alpha_0,\alpha_1,\dots,\alpha_{I-1}$ with $\sum_{i=0}^{I-1}\alpha_i=1$ so that $\sum_{i=0}^{I-1}\alpha_i Q_i(x) = \pi_x$ holds for all $x\in\mathcal{X}$. Let $\boldsymbol{\alpha} \triangleq [\alpha_0,\alpha_1,\dots,\alpha_{I-1}]^T$. Since the set $Q_0,Q_1,\dots,Q_{I-1}$ is linearly dependent, it follows that $A=[Q_0,Q_1,\dots,Q_{I-1}]$ has a nonempty null space. Moreover, since $Q_i$'s are probability distribution vectors, a vector in the null space of $A$ must have both positive and negative entries. Let us take an element $\boldsymbol{v}=[v_0,v_1,\dots,v_{I-1}]^T$ in the null space, normalized such that $\sum_{i=0}^{I-1}v_i=1$ holds. By using the vector $\boldsymbol{v}$, we obtain a new set of coefficients for which the Bayesian plausibility constraint is satisfied. More specifically, we let 
\begin{align*}
\boldsymbol{\lambda}  
= \theta \boldsymbol{\alpha}
+ (1-\theta) \boldsymbol{v}
\end{align*}
where $\boldsymbol{\lambda} = [\lambda_0,\lambda_1,\dots,\lambda_{I-1}]^T$, and $\theta\in(0,1)$ is chosen such that $\lambda_i=0$ for some $i\in\{0,1,\dots,I-1\}$ and $\lambda_j\geq 0$ for all $j\neq i$. It is possible to find such a $\theta$ since $\boldsymbol{v}$ has both positive and negative entries. Furthermore, by using a normalized vector $\boldsymbol{v}$, we always get $\sum_{i=0}^{I-1}\lambda_i=1$. This reveals that $\boldsymbol{\lambda}$ is a valid probability distribution vector. Hence, it is possible to satisfy the condition in \eqref{eq:algo_eq3} by using the set of coefficients  $\lambda_0,\lambda_1,\dots,\lambda_{I-1}$ as probabilities corresponding to the posteriors $Q_0,Q_1,\dots,Q_{I-1}$. In addition, $\boldsymbol{\lambda}$ has fewer nonzero entries than the probability distribution vector $\boldsymbol{\alpha}$. This allows us to write 
\begin{align*}
\boldsymbol{\alpha} 
=\theta 
\boldsymbol{\lambda}
+ 
(1-\theta)
\boldsymbol{\gamma},
\end{align*}
where $\boldsymbol{\gamma}=[\gamma_0,\gamma_1,\dots,\gamma_{I-1}]^T$, and $\theta\in(0,1)$ is chosen such that $\gamma_i=0$ for some $i\in\{0,1,\dots,I-1\}$ and $\gamma_j\geq 0$ for all $j\neq i$. By this choice of $\theta$, we get a feasible probability distribution vector $\boldsymbol{\gamma}$. Thus, it is seen that the probability distribution vector $\boldsymbol{\alpha}$ can be expressed as a convex combination of probability distribution vectors $\boldsymbol{\lambda}$ and $\boldsymbol{\gamma}$ that have fewer nonzero terms than $\boldsymbol{\alpha}$. In other words, the initial convex combination vector can be represented as a convex combination of two distinct convex combination vectors that have fewer posteriors in the combination. If the resulting posterior combinations are still linearly dependent, we can apply the same technique repeatedly until each combination involves linearly independent posteriors. As a result, we obtain
\begin{align}
\boldsymbol{\alpha} = \theta_0 \boldsymbol{\lambda}_0+\theta_1 \boldsymbol{\lambda}_1\dots+\theta_n \boldsymbol{\lambda}_n
\label{eq:thm_algo_eq1}
\end{align}
for some $n\geq 2$, where $\boldsymbol{\lambda}_i = [\lambda_i^0,\lambda_i^1,\dots,\lambda_i^{I-1}]^T$, $\sum_{j=0}^{I-1}\lambda_i^jQ_j(x)=\pi_x$ for all $x\in\mathcal{X}$ and $i\in\{0,1\dots,n\}$, $\sum_{i=0}^n\theta_i=1$, and $\theta_i>0$ for all $i\in\{0,1\dots,n\}$. In \eqref{eq:thm_algo_eq1}, for each $\boldsymbol{\lambda}_i$, the set of posteriors corresponding to nonzero entries is a linearly independent set. Notice that $\boldsymbol{\lambda}_i$'s are feasible probability distribution vectors with an expected payoff value of 
\begin{align*}
U_i^s = \sum_{j=0}^{I-1} u^{s,*}(Q_j) \lambda_i^j. 
\end{align*}
If $U_i^s=U_{i'}^s$ holds for all $i\neq i'$, then the expected payoff with the convex combination vector $\boldsymbol{\alpha}$ becomes the same. On the other hand, if there exist $i,i'\in\{0,1,\dots,n\}$ such that $U_i^s\neq U_{i'}^s$ holds, then the expected payoff value with $\boldsymbol{\alpha}$ cannot be optimal. That is, by using $\boldsymbol{\lambda}_i$ with the largest $U_i^s$, we obtain a better payoff than the payoff with $\boldsymbol{\alpha}$. This proves that it is sufficient to consider only the posterior combinations that are linearly independent. Finally, the algorithm compares the expected payoffs in \eqref{eq:algo_eq4} for each feasible and linearly independent posterior combination to find the optimal payoff value.  

\subsection{Proof of Lemma~\ref{lem:fully_revealing_lemma}}\label{proof:fully_revealing_lemma}

Under a fully revealing policy, the receiver takes the action $y=i$ when the source observation is $x=i$ for all $i=0,1,\dots,M-1$. It is possible to convince the receiver to take an action $y=j$ for a source observation $x=i$ via quantization when $\pi_j>\pi_i$ holds. In particular, one can quantize the source values $x=i$ and $x=j$ into a single quantization bin while revealing the rest of the source values. In this case, the receiver takes the action $y=j$ for this bin. This policy outperforms a fully revealing policy if 
\begin{align*}
\pi_i u^s(i,j) + \pi_j u^s(j,j)>\pi_i u^s(i,i) + \pi_j u^s(j,j)
\end{align*}
holds, where the left and the right hand sides correspond to the payoff contributions due to these sources under the described quantization policy and a fully revealing policy, respectively. This implies that a fully revealing policy cannot be optimal when $\pi_j>\pi_i$ and $u^s(i,j) >u^s(i,i)$ are satisfied. On the other hand, if $\pi_i=\pi_j$ is satisfied, the receiver is indifferent between the actions $y=i$ and $y=j$ under a quantization policy mapping the sources $x=i$ and $x=j$ to the same message while revealing the rest of the sources. In this case, the receiver selects the action that the sender prefers. Therefore, such a quantization policy outperforms a full revealing policy if the following holds:
\begin{align*}
&\max\{ \pi_iu^s(i,i)+\pi_ju^s(j,i), \pi_iu^s(i,j)+\pi_ju^s(j,j)\}\\
&> \pi_i u^s(i,i) + \pi_j u^s(j,j).
\end{align*}
Hence, for equal prior sources $x=i$ and $x=j$, a fully revealing policy cannot be the equilibrium solution if at least one of the following conditions is satisfied: $u^s(j,i)>u^s(j,j)$ and $u^s(i,j)>u^s(i,i)$. Then, it is concluded that the conditions in \eqref{eq:lem_fully_reveal_eq1} and \eqref{eq:lem_fully_reveal_eq2} are necessary for the optimality of a fully revealing policy. Finally, we show that the conditions in \eqref{eq:lem_fully_reveal_eq1} and \eqref{eq:lem_fully_reveal_eq2} are sufficient for the optimality of full revelation. We know that the sender cannot convince the receiver to take an action $y=j$ under a source observation of $x=i$ by using a deterministic encoding policy when $\pi_j<\pi_i$ is satisfied. On the other hand, for $\pi_j> \pi_i$, the quantization policy described above leads to $y=j$ under a source observation of $x=i$. Similarly, for $\pi_j= \pi_i$, such a quantization policy leads to either $y=j$ for $x=i$, or $y=i$ for $x=j$. Nevertheless, such a quantization policy does not lead to a higher payoff for the sender under the conditions in \eqref{eq:lem_fully_reveal_eq1} and \eqref{eq:lem_fully_reveal_eq2}. Moreover, any other quantization policy further reduces the sender's payoff by inducing a receiver action that is not equal to the source for more than one source. Therefore, it is optimal for the sender to reveal the source completely under the conditions in \eqref{eq:lem_fully_reveal_eq1} and \eqref{eq:lem_fully_reveal_eq2}.

\subsection{Proof of Theorem~\ref{thm:fully_revealing}}\label{proof:fully_revealing_thm}

Due to the receiver's payoff structure, for any given $i,j\in\mathcal{X}$, a posterior $Q$ with $Q(x)=0$ for all $x\neq i$ and $x\neq j$ induces either $y=i$ or $y=j$ as the best response of the receiver. This implies that we can treat each source pair with the corresponding action pair separately. Then, if pairwise randomization improves the sender's payoff for at least one pair of source values, then one can improve the sender's overall payoff by applying randomization for that pair while fully revealing other source realizations. In the following, we assume that pairwise randomization does not improve the sender's payoff for any source pair and prove that the equilibrium solution involves a sender that completely reveals the source in this case. For $i\in\mathcal{Y}$, let $Q_i$ denote the posterior with $Q_i(i)=1$. For $i,j\in\mathcal{Y}$, let $Q_{i,j}$ denote the critical posterior under which the receiver prefers $y=i$ and $y=j$ equally, i.e., $Q_{i,j}(i)=Q_{i,j}(j)=1/2$. In the following, we show that 
\begin{align}
u^{s,*}(Q_{i,j}) 
\leq (1/2) u^{s,*}(Q_i) 
+ (1/2) u^{s,*}(Q_j)
\label{eq:thm_fully_reveal_eq1}
\end{align}
must hold when the sender does not wish to apply randomization to the source pair $x=i$ and $x=j$. In order to prove \eqref{eq:thm_fully_reveal_eq1}, we assume otherwise and then reach a contradiction. Theorem~\ref{thm:algorithm} states that it is sufficient to consider the posteriors $Q_i$, $Q_j$ and $Q_{i,j}$ when we isolate the source pair $x=i$ and $x=j$. The posteriors $Q_i$ and $Q_j$ are induced with probabilities $\pi_i$ and $\pi_j$ when the sender uses a fully revealing policy. If the inequality in \eqref{eq:thm_fully_reveal_eq1} does not hold, then it is possible to increase the sender's payoff by inducing $Q_{i,j}$ with a certain nonzero probability. In particular, let us consider a family of encoding policies parametrized by $\alpha$ that leads to payoff values of
\begin{align*}
U^s(\alpha)
&= (\pi_i-\alpha/2)u^{s,*}(Q_i) \nonumber \\
&+ (\pi_j-\alpha/2)u^{s,*}(Q_j) + \alpha  u^{s,*}(Q_{i,j}),
\end{align*}
where $0\leq \alpha\leq 1$ is such that $\pi_i-\alpha /2\geq 0$ and $\pi_j-\alpha /2 \geq 0$ are satisfied. If we take $\alpha=0$, we obtain the payoff value under a fully revealing policy. Since we assume that the inequality in \eqref{eq:thm_fully_reveal_eq1} does not hold, the payoff value strictly increases when we increase $\alpha$. In addition, we obtain the highest payoff by taking $\alpha$ as the largest value satisfying $\pi_i-\alpha/2\geq 0$, $\pi_j-\alpha/2\geq 0$ and $0\leq \alpha\leq 1$. We also know that for the source pair $x=i$ and $x=j$, a noninformative policy cannot outperform a fully revealing policy since we assume that the equilibrium solution under a deterministic policy restriction is fully revealing. This ensures that the resulting policy with the largest $\alpha$ cannot be equivalent to a noninformative policy. Therefore, by applying randomization, we obtain a payoff value that is higher than the fully revealing payoff under the assumption that \eqref{eq:thm_fully_reveal_eq1} does not hold. This is a contradiction to our assumption that pairwise randomization does not improve the sender's payoff. Hence, it follows that \eqref{eq:thm_fully_reveal_eq1} must hold. Moreover, since $i$ and $j$ are arbitrary, it follows that for each source pair, there is no need to induce the posterior under which the receiver is indifferent between the corresponding pair of actions. 

Next, we prove that for each action combination with more than two elements, the sender does wish to induce the critical posterior under which the receiver prefers all the actions in the combination equally. From \eqref{eq:thm_fully_reveal_eq1}, it follows that $u^s(i,j)\leq u^s(i,i)$ and $u^s(j,i)\leq u^s(j,j)$ must hold for all $i,j\in\mathcal{Y}$ satisfying $i\neq j$. Let us take an action combination $\mathcal{A}\subseteq \mathcal{Y}$. Let $Q$ denote the critical posterior under which the receiver is indifferent between the actions in $\mathcal{A}$, i.e., $Q(x)=1/|\mathcal{A}|$ for all $x\in\mathcal{A}$ and $Q(x)=0$ for all $x\notin\mathcal{A}$. Then, we can write
\begin{align*}
u^{s,*}(Q)=\max_{y\in\mathcal{A}}\left\{\sum_{x\in\mathcal{X}} u^s(x,y)Q(x) \right\} 
\leq 
\sum_{x\in\mathcal{X}} u^s(x,x) Q(x),
\end{align*}
where the inequality follows from the fact that $u^s(x,y)\leq u^s(x,x)$ holds for all $x,y\in\mathcal{Y}$, and the right hand side corresponds to the payoff if $[1,0,\dots,0]^T$, $\dots$, $[0,\dots,0,1]^T$ are induced instead of $Q$. Since $\mathcal{A}\subseteq \mathcal{Y}$ is arbitrary, it follows that it is optimal for the sender to induce only the posteriors with a single nonzero element. This implies that the equilibrium solution involves a fully revealing sender.

\subsection{Proof of Lemma~\ref{lem:noninf_lemma}}\label{proof:noninf_lemma}

Under a noninformative policy, the receiver takes the action $y=0$ for all sources. If the sender maps a source observation of $x=i$ for $i\in\{1,2,\dots,M-1\}$ to a distinct message, then the receiver takes the action $y=i$ for this source. Then, if $u^s(i,i)>u^s(i,0)$ holds, a quantization policy that maps only the source $x=i$ to a distinct message outperforms a noninformative policy. Therefore, the condition in \eqref{eq:noninf_lemma_eq1} is necessary for the optimality of a noninformative policy. In addition, the receiver takes the action $y=i$ for source observations of $x=i$ and $x=j$ with $i,j\in \{1,2,\dots,M-1\}$ and $i<j$ when the sender maps only these source realizations to the same message. Thus, a quantization policy that conveys whether the source belongs to the set $\{i,j\}$ outperforms a noninformative policy if $\pi_iu^s(i,i)+\pi_ju^s(j,i)>\pi_iu^s(i,0)+\pi_ju^s(j,0)$ is satisfied. Hence, the optimality of a noninformative policy requires \eqref{eq:noninf_lemma_eq2}. Next, we show that the inequalities in \eqref{eq:noninf_lemma_eq1} and \eqref{eq:noninf_lemma_eq2} are sufficient conditions. We know that the sender may convince the receiver to take the action $y=i$ under a source observation of $x=j$ by using a deterministic policy only when $i\leq j$ holds. This can be achieved by using a quantization policy described above. Nevertheless, under \eqref{eq:noninf_lemma_eq1} and \eqref{eq:noninf_lemma_eq2}, such a policy does not improve the sender's payoff compared to the noninformative payoff. In addition, the inequalities in \eqref{eq:noninf_lemma_eq1} and \eqref{eq:noninf_lemma_eq2} can be applied iteratively for any other quantization policy to conclude that a noninformative policy is optimal.

\subsection{Proof of Theorem~\ref{thm:noninf}}\label{proof:noninf_thm}

We show that if the condition in \eqref{eq:noninf_thm_cond} is satisfied for a critical posterior inducing equal receiver preferences, then it is possible to improve the sender's payoff by applying randomization. Let $\mathcal{A}\subseteq\mathcal{Y}$ denote a set of actions with cardinality $n\geq 2$. The critical posterior distribution $Q$ under which the receiver is indifferent between the actions in $\mathcal{A}$ is given by $Q(x)=1/n$ for all $x\in\mathcal{A}$, and $Q(x)=0$ for all $x\notin\mathcal{A}$. Suppose that the sender prefers $y=i$ under the posterior distribution $Q$ for some $i\in\mathcal{A}$ with $i\neq 0$. A noninformative policy is equivalent to inducing $\pi=[\pi_0,\pi_1,\dots,\pi_{M-1}]^T$ as the only posterior distribution at the receiver. Such a noninformative policy yields a payoff value of $u^{s,*}(\pi)=\sum_{x\in\mathcal{X}}\pi_x u^s(x,0)$. Consider an alternative encoding policy that leads to the following posterior distributions at the receiver: $Q$ and $Q' = (\pi-\alpha Q)/(1-\alpha)$ with $P(Q)=\alpha$ and $P(Q')=1-\alpha$, where $\alpha>0$ is small to to ensure that $Q'$ is a valid probability distribution, and that the receiver decides $y=0$ under $Q'$. Then, we compare the payoffs in the following: 
\begin{align}
&u^{s,*}(\pi) 
=\sum_{x\in\mathcal{X}} (\alpha Q(x)+(1-\alpha)Q'(x)) u^s(x,0) \nonumber \\
&<\alpha \sum_{x\in\mathcal{X}} u^s(x,i) Q(x) + 
(1-\alpha) \sum_{x\in\mathcal{X}}Q'(x)) u^s(x,0) \nonumber \\
&= \alpha u^{s,*}(Q) + (1-\alpha)u^{s,*}(Q'),
\label{eq:noninf_thm_eq1}
\end{align}
where the inequality follows from \eqref{eq:noninf_thm_cond}, and the last equality is due to the fact that the receiver selects $y=i$ and $y=0$ under $Q$ and $Q'$, respectively. From \eqref{eq:noninf_thm_eq1}, it is seen that a higher payoff is obtained by assigning a nonzero probability to a critical posterior inducing equal receiver preferences for certain actions. This reveals that the sender's payoff improves via randomization. 

Next, we show that if the inequalities in \eqref{eq:noninf_thm_cond} are not satisfied for all critical posteriors, then the equilibrium solution involves a noninformative encoder. In particular, for any critical posterior $Q$ inducing equal receiver preferences for certain actions, we have that   
\begin{align}
\sum_{x\in\mathcal{X}}u^s(x,0)Q(x) \geq  u^{s,*}(Q).
\label{eq:noninf_thm_eq2}
\end{align}
We note that if the sender induces a critical posterior under which the receiver is indifferent between a set of actions including $y=0$, then the receiver decides $y=0$ for this posterior due to \eqref{eq:noninf_thm_eq2}. This implies that in order to see whether randomization improves the sender's payoff, it is sufficient to consider the critical posteriors under which the receiver is indifferent between a set of actions that does not contain $y=0$. In the following, we suppose that the sender induces such posteriors with nonzero probabilities and then show that the payoff in this case cannot outperform the noninformative payoff. In particular, let $\{Q_i\}_{i=0}^n$ denote a set of critical posteriors that are induced with a nonzero probability for a given encoding policy. Such an encoding policy must satisfy
\begin{align}
\pi = \sum_{i=0}^n \alpha_i Q_i + \sum_{i=1}^{M-1} \beta_i e_i + \gamma Q,
\label{eq:noninf_thm_eq3}
\end{align}
where $e_i$'s denote posteriors with $e_i(i)=1$, $\alpha_i\triangleq P(Q_i)$, $\beta_i\triangleq P(e_i)$, $\gamma\triangleq P(Q)$, and $Q$ is a posterior so that the convex combination of the induced posteriors is equal to the prior. We know that the largest entry in $\pi$ is $\pi_0$. In addition, $Q_i(0)=0$ holds for all $i\in\{0,1,\dots,n\}$ and $e_i(0)=0$ holds for all $i\in\{1,2,\dots,M-1\}$. These imply that $Q(0)>Q(i)$ is satisfied for all $i\in\{1,2,\dots,M-1\}$. Therefore, the receiver decides $y=0$ under the posterior distribution $Q$. In the following, we iteratively combine $Q_i$'s with $Q$ to obtain a larger or equal payoff. More specifically, in the first step, instead of inducing $Q_0$ and $Q$ with probabilities $\alpha_0$ and $\gamma$, a posterior of $(\alpha_0 Q_0 +\gamma Q)/(\alpha_0+\gamma)$ is 
induced with a probability of $(\alpha_0+\gamma)$. This does not reduce the sender's payoff since 
\begin{align*}
&(\alpha_0+\gamma) u^{s,*}((\alpha_0 Q_0+\gamma Q)/(\alpha_0+\gamma))
\\
&\leq \alpha_0 u^{s,*}(Q_0) 
+\gamma u^{s,*}(Q)
\end{align*}
holds due to the facts that \eqref{eq:noninf_thm_eq2} is satisfied for $Q_0$ and that the receiver still decides $y=0$ under $(\alpha_0 Q_0 +\gamma Q)/(\alpha_0+\gamma)$. We can repeat this procedure to obtain an improved payoff value of 
\begin{align*}
\sum_{i=1}^{M-1}\beta_iu^{s,*}(e_i) + \gamma' u^{s,*}(Q') 
\end{align*}
where $Q'$ is the resulting posterior after combining $Q_i$'s and $Q$ in \eqref{eq:noninf_thm_eq3}, and $\gamma'\triangleq P(Q')$. Finally, we can iteratively combine $e_i$'s with $Q'$ to further increase the payoff value. More specifically, in the first step, we use $(e_1+Q')/(\beta_1+\gamma')$ with probability $(\beta_1+\gamma')$ instead of $e_1$ and $Q'$ with probabilities $\beta_1$ and $\gamma'$. Then, one can show that this does not reduce the sender's payoff since $u^s(i,0)\geq u^s(i,i)$ holds due to Lemma~\ref{lem:noninf_lemma}. We can repeat this procedure for all $e_i$ with a nonzero probability to obtain a larger or equal payoff. As a result, we obtain a policy that is equivalent to a noninformative policy. Hence, if the inequalities in \eqref{eq:noninf_thm_cond} are not satisfied, then the equilibrium solution involves a noninformative encoder. 

\subsection{Proof of Lemma~\ref{lem:quantized}}\label{proof:quantized_lemma}

We can apply Lemma~\ref{lem:noninf_lemma} for each quantization bin to obtain the conditions in (i) and (ii). If the conditions in (iii) are not met, we can change the bin of the source $x$ from $B_i$ to $B_j$ to improve the sender's payoff. In particular, since such a change does not affect the receiver's decision for sources other than $x$, the sender's payoff improves by modifying the bins under the condition of $u^s(x,x_i)<u^s(x,x_j)$. Moreover, we can construct a separate bin from sources $x$ and $y$ to increase the sender's payoff if the corresponding condition in (iv) is not met. Next, for a pair of source values $x_i$ and $x_j$ with $x_i<x_j$, the sender may convince the receiver to take the action $x_i$ under a source observation of $x_j$. However, since $x_j$ is the receiver's decision for the bin $B_j$, it is required to take the decisions for the other sources in $B_j$ into account. In (v), one searches for a modified quantization policy that only affects the decisions for the sources in $\{x_i\}\cup B_j$. Therefore, if the search leads to a quantization policy other than with bins $\{x_i\}$ and $B_j$, then the quantization policy with bins $B_0,B_1,\dots,B_{n-1}$ cannot be optimal. All in all, we have shown that the conditions in (i)-(v) are necessary for the optimality of a quantization policy with bins $B_0,B_1,\dots,B_{n-1}$. Finally, it remains to show that these conditions are sufficient, as well. These conditions ensure that the sender does not wish to induce a different feasible action for any source with minimal effect on other decisions. This implies that any other quantization policy further reduces the sender's payoff compared to the individual payoffs considered in the analysis above.

\subsection{Proof of Theorem~\ref{thm:quantized}}\label{proof:quantized_thm}

To begin with, we can apply Theorem~\ref{thm:noninf} to obtain the result in (i). In particular, a quantization policy can be improved by applying randomization to sources within a bin if the corresponding condition in \eqref{eq:quantized_thm_eq1} is satisfied. Furthermore, due to Theorem~\ref{thm:noninf}, these conditions are necessary and sufficient for the optimality of a randomized policy when the randomization is restricted to combinations of sources inside a bin. However, one can still improve the sender's payoff by applying randomization to source combinations from multiple bins. In the following, we consider such randomized policies under the assumption that the conditions in (i) are not satisfied.

Next, let us take a posterior distribution $Q$ that induces equal receiver preferences for actions corresponding to sources from multiple bins. The aim is to see whether inducing such a posterior improves the sender's payoff compared to the quantization policy that is optimal among deterministic policies. This quantization policy induces a single posterior for each bin. In order to induce $Q$ with a nonzero probability, one needs to decrease the probabilities of posteriors for the optimal quantization policy in such a way that the Bayesian plausibility constraint is still satisfied. In the following, we write $Q$ as a linear combination of certain posteriors to increase its probability while satisfying the Bayesian plausibility constraint. Towards that goal, let us consider the bin $B_i$ for some $i\in\{0,1,\dots,n-1\}$. We can express the posterior of the bin $B_i$ for the optimal quantization policy as a convex combination of critical posteriors inside the region that the receiver prefers $x_i$. Since the conditions in (i) are not satisfied, the sender prefers $x_i$ under these posteriors. This implies that the receiver's decision becomes $y=x_i$ for these posteriors. The matrix $R\triangleq [Q_0,Q_1,\dots,Q_{M-1}]$ in \eqref{eq:quantized_thm_eq3} contains a subset of such posteriors. In particular, only $Q_i$'s with $Q_i(x_i)=1$ and the posteriors under which the receiver is indifferent between $x_i$ and some other action in $B_i$ are considered. In fact, any critical posterior inducing equal receiver preferences for more than two actions including $x_i$ within the bin $B_i$ can be written as a linear combination of certain posteriors in $R$ that lead to $x_i$ as the best response of the receiver. This implies that it is sufficient to consider the posteriors in $R$ while increasing the probability of $Q$. Since $R$ is full rank by construction, it is possible to take its inverse. Hence, we compute \eqref{eq:quantized_thm_eq3} to express $Q$ as $Q=\sum_{i=0}^{M-1}\alpha_i Q_i$. Then, if the corresponding condition in \eqref{eq:quantized_thm_eq2} is satisfied, it is possible to increase the sender's payoff via randomization by making the following modifications to the optimal quantization policy. We first write the posterior of each bin as a convex combination of critical posteriors inside the region that the receiver prefers $x_i$, where we take nonzero probabilities in the combination. Then, the probability of the posterior $Q_i$ in $R$ is increased (decreased) with probability $\alpha_i \epsilon$ for positive (negative) $\alpha_i$, where $\epsilon>0$ is such that the resulting probabilities are nonnegative. This allows us to assign a probability of $\epsilon$ for the posterior $Q$. Such a randomized policy outperforms the optimal quantization policy since \eqref{eq:quantized_thm_eq2} is satisfied. On the other hand, if the conditions in \eqref{eq:quantized_thm_eq2} are not satisfied for all such posteriors, it is not possible to increase the sender's payoff via randomization. To see this, we note that $\alpha_i$'s in \eqref{eq:quantized_thm_eq2} are unique since $R$ is full rank. Hence, the only way to increase the probability of the posterior $Q$, a positive factor of $\alpha_i$'s in \eqref{eq:quantized_thm_eq3} must be used while increasing/decreasing the corresponding probabilities of $Q_i$'s. However, this does not improve the sender's payoff since \eqref{eq:quantized_thm_eq2} does not hold.

\bibliographystyle{IEEEbib}
\bibliography{BayPers}


\end{document}